\documentclass[12pt]{article}
\usepackage{bm}
\usepackage{amsmath}
\usepackage{amssymb, amscd, amsthm,amsfonts}
\usepackage[dvips]{graphicx}
\usepackage{verbatim}
\usepackage[dvips]{color}
\newtheorem{theorem}{Theorem}
\newtheorem{lemma}{Lemma}

\newtheorem{corollary}{Corollary}

\usepackage{pict2e}

\begin{document}

\title{On the Synthetic Channels in Polar Codes over Binary-Input Discrete Memoryless Channels\footnote{This work was supported by the Natural Science Foundation of China (No. 61977056).}}
\author{Yadong Jiao, Xiaoyan Cheng, Yuansheng Tang\footnote{Corresponding author.}\\
{\it\small School of Mathematical Sciences, Yangzhou University, Jiangsu, China}\\
and Ming Xu\footnote{Email addresses: dx120210046@stu.yzu.edu.cn(Y. Jiao), xycheng@yzu.edu.cn(X. Cheng), ystang@yzu.edu.cn(Y. Tang), mxu@szcu.edu.cn(M. Xu)}
\\
{\it \small Suzhou City University, Jiangsu, China}}
\date{}
\maketitle

\begin{abstract}
Polar codes introduced by Arikan in 2009 are the first code family achieving the capacity of binary-input discrete memoryless channels (BIDMCs) with low-complexity encoding and decoding. Identifying unreliable synthetic channels in polar code construction is crucial. Currently, because of the large size of the output alphabets of synthetic channels, there is no effective approach to evaluate their reliability, except in the case that the underlying channels are binary erasure channels. This paper defines equivalence and symmetry based on the likelihood ratio profile of BIDMCs and characterizes symmetric BIDMCs as random switching channels (RSCs) of binary symmetric channels.
By converting the generation of synthetic channels in polar code construction into algebraic operations on underlying channels, some compact representations of RSCs for these synthetic channels are derived. Moreover, a lower bound for the average number of elements that possess the same likelihood ratio within the output alphabet of any synthetic channel generated in polar codes is also derived.
\end{abstract}
\vspace{1ex}
{\noindent\small{\bf Keywords:}
    	Polar code; BIDMC; synthetic channel; likelihood ratio profile; random switching channel}

\section{Introduction}

Polar codes introduced by Arikan in 2009 are the first code family achieving the capacity of binary-input discrete memoryless channels (BIDMCs) under successive cancellation (SC) decoding as the code length tends to infinity. Due to their low complexity encoding and decoding, polar codes were selected in 2016 by the 3GPP Group for the uplink/downlink channel control in the 5G standard.

By using the kernel matrix $G$ of order 2, one can transform two independent BIDMCs into two synthetic channels (cf. \cite{Arikan09}) that are equivalents of the parity-constrained-input parallel channel and the parallel broadcast channel defined in \cite{Shamai05}. These synthetic channels, which and their compounds are called Arikan transformations of the underlying channels in this paper, preserve the sum of capacities.
The idea of polar codes proposed in \cite{Arikan09} is as follows. At first, $N=2^k$ independent
copies of the underlying BIDMC are iteratively transformed $k$ times into $N$ synthetic channels
by using the kernel matrix $G$. When $N$ is large enough, the synthetic channels
polarize into a set of reliable channels and a set of unreliable channels. Then, reliable communication is realized by transmitting the information bits only on the reliable channels while the remaining channels are {\it frozen} (i.e., their inputs are transparent for the receiver). Clearly, for polar codes, it is critical to identify the unreliable synthetic channels.

About the determination of the exact reliability for all the synthetic channels generated in polar codes, Arikan proposed in \cite{Arikan09} a recursive algorithm based on the Bhattacharyya parameter. However, this algorithm is efficient only for the case that the underlying channel is a binary erasure channel (BEC).
In general, the ranking of the synthetic channels depends
on the underlying channel. In literature, several techniques have been proposed to estimate the reliability of the synthetic channels: Monte
Carlo simulation \cite{Arikan09}, density evolution and its Gaussian approximation
\cite{Tanaka09}, \cite{Trifonov12}, efficient degrading
and upgrading method \cite{Vardy13}, polarization weight and $\beta$-expansion \cite{3GPP16}, \cite{He17}, etc.
Meanwhile, it was observed that
there is a partial order (with respect to degradation) between
the synthetic channels, which holds for any underlying
channel. By exploiting the partial order, the
complexity of the code construction can be significantly
reduced (cf. \cite{Urbanke19}, \cite{Siegel19}).

In this paper, our main focus is on providing the most compact representations for the synthetic channels generated in polar codes over symmetric BIDMCs. Through these representations, the reliability of the synthetic channels can be computed efficiently.

The paper is organized as follows.
In Section~\ref{sec02}, we introduce a few definitions and some elementary properties of likelihood ratio profile, equivalence, random switching, symmetry and degradation of BIDMCs.
In Section~\ref{sec03}, we introduce some elementary properties of Arikan transformations of arbitrary BIDMCs and some algebraic operations between the Arikan transformations of symmetric BIDMCs.
In Section~\ref{sec04}, we give a brief introduction for the polar codes from a new angle of view and a method for representing the Arikan transformations of symmetric BIDMCs as
random switchings of binary symmetric channels (BSCs). A lower bound for the average number of elements that possess the same likelihood ratio within the output alphabet of any synthetic channel generated in polar codes is also derived in this section.
Finally, conclusions are drawn in Section \ref{sec05}.
\section{Binary-Input Discrete Memoryless Channels}
\label{sec02}
In this paper, for any random variable $v$ the notation $v$ may also express, a concrete value in its sample space, or the probability event that $v$ takes a concrete value, if there is no confusion.
For example, $\Pr(v)$ denotes the probability of that $v$ takes a concrete value which is also denoted by $v$.
In particular, if $u$ is a random variable whose sample space is a subset of that of $v$, then $v=u$ expresses the probability event that $v$ takes a concrete value denoted by $u$.

Let $W:x\in\mathcal{X}\mapsto y\in\mathcal{Y}$ be a \emph{binary-input discrete memoryless channel} (BIDMC), where the input $x$ is always supposed to be {\it uniformly distributed} in $\mathcal{X}=\mathbb{F}_2=\{0,1\}$, the finite field of two elements, and the output alphabet $\mathcal{Y}$ is a discrete set.

The {\it symmetric capacity} of $W$ is defined as
\begin{equation}\label{sym_cap} I(W)=\sum_{y\in\mathcal{Y}}\sum_{x\in\mathcal{X}}\frac{1}{2}\Pr(y|x)
\log_2\frac{\Pr(y|x)}{\frac{1}{2}\Pr(y|x=0)+\frac{1}{2}\Pr(y|x=1)},
\end{equation}
which is equal to the {\it mutual information}
$I(x;y)=H(x)+H(y)-H(x,y)$
between the output $y$ and the input $x$.

The \emph{likelihood ratio} of $\hat{y}\in\mathcal{Y}$ is defined as
$\mathcal{L}_W(\hat{y})=\Pr(y=\hat{y}|x=0)/\Pr(y=\hat{y}|x=1)$.
The \emph{maximum likelihood decoding} (MLD) of the BIDMC $W$ decodes $\hat{y}\in\mathcal{Y}$ into $d_{\text{mld}}(\mathcal{L}_W(\hat{y}))$, where
$d_{\text{mld}}(l)\in\mathcal{X}$ equals
0 if $l\geq 1$ and 1 otherwise.
Then, the {\it probability of error decoding} for MLD of $W$ is given by
\begin{equation}\label{error probability}
	P_{\epsilon}(W)=\frac{1}{2}\sum_{y\in\mathcal{Y}}\min\{\Pr(y|x=0),\Pr(y|x=1)\}.
\end{equation}

The reliability of $W$ can also be scaled by the Bhattacharyya parameter
\begin{gather}
Z(W)=\sum_{y\in\mathcal{Y}}\sqrt{\Pr(y|x=0)\Pr(y|x=1)}.
\label{Bp}
\end{gather}

\subsection{Equivalence of BIDMCs}
For $\varepsilon\in[0,1]$, let $\overline{\varepsilon}$ denote $1-\varepsilon$ and $L_W(\varepsilon)$ the following subset of $\mathcal{Y}$
\begin{align}
L_W(\varepsilon)=\{y\in\mathcal{Y}: \varepsilon\Pr(y|x=0)=\overline{\varepsilon}\Pr(y|x=1)\}.
\end{align}
Clearly, $\{L_W(\varepsilon)\}_{\varepsilon\in[0,1]}$ is a partition of $\mathcal{Y}$ determined by distinguishing the likelihood ratios.
Since the input $x$ of $W$ is uniformly distributed in $\mathcal{X}$, the output $y$ of $W$ belongs to the set $L_W(\varepsilon)$ with probability
\begin{align}
\Pr(y\in L_W(\varepsilon))=P_W(\varepsilon)=\frac{1}{2}\sum_{x\in\mathcal{X},\,y\in L_W(\varepsilon)}\Pr(y|x).
\label{rr01}
\end{align}
The function $P_W(\varepsilon)$ over $[0,1]$ is also called the {\it likelihood ratio profile} (LRP) of $W$ in this paper.

Since for any BIDMC the likelihood ratio constitutes a sufficient
statistic with respect to decoding (c.f. \cite{Urbanke08}), two BIDMCs $W$, $W'$ are said to be \emph{equivalent}, and written as $W\cong W'$, if their LRPs are the same, i.e.,
$P_{W}(\varepsilon)=P_{W'}(\varepsilon)$ is valid for all $\varepsilon\in[0,1]$.
Indeed, we have the following theorem.
\begin{theorem}\label{lem00s}
For any BIDMC $W$, we have \begin{gather}
I(W)=1-\sum_{\varepsilon\in[0,1]} \hbar(\varepsilon)P_W(\varepsilon),\label{f02}\\
P_{\epsilon}(W)=\sum_{\varepsilon\in[0,1]}\min\{\varepsilon,\overline{\varepsilon}\} P_W(\varepsilon),\label{f01}\\
Z(W)=2\sum_{\varepsilon\in[0,1]}\sqrt{\varepsilon\overline{\varepsilon}} P_W(\varepsilon),\label{f03}
\end{gather}
where $\hbar(\varepsilon)=-\varepsilon\log_2 \varepsilon-\overline{\varepsilon}\log_2 \overline{\varepsilon}$ is the binary entropy function.
Therefore, for any BIDMC\ $W'$ with $W'\cong W$, we have $I(W')=I(W)$, $P_{\epsilon}(W')=P_{\epsilon}(W)$ and $Z(W')=Z(W)$.
\end{theorem}
\begin{proof}
Clearly, for any $\varepsilon\in[0,1]$ and $y\in L_W(\varepsilon)$, we have
\begin{align}
\Pr(y|x=1)=\varepsilon\sum_{x\in\mathcal{X}}\Pr(y|x),\
\Pr(y|x=0)=\overline{\varepsilon}\sum_{x\in\mathcal{X}}\Pr(y|x).\label{d01}
\end{align}
From (\ref{sym_cap}), (\ref{rr01}) and (\ref{d01}) we see
\begin{align*}
I(W)&=\sum_{\varepsilon\in[0,1]}\sum_{y\in L_W(\varepsilon)}\frac{1}{2}\Big(\Pr(y|x=1)\log_2 (2\varepsilon)+\Pr(y|x=0)\log_2 (2\overline{\varepsilon})\Big)\nonumber\\
&=\sum_{\varepsilon\in[0,1]}\sum_{y\in L_W(\varepsilon)}\frac{1}{2}\Big(\varepsilon\log_2 (2\varepsilon)+\overline{\varepsilon}\log_2 (2\overline{\varepsilon})\Big)\sum_{x\in\mathcal{X}}\Pr(y|x)\nonumber\\
&=\sum_{\varepsilon\in[0,1]}(\varepsilon\log_2 (2\varepsilon)+\overline{\varepsilon}\log_2 (2\overline{\varepsilon}))P_W(\varepsilon)\nonumber\\
&=1-\sum_{\varepsilon\in[0,1]} \hbar(\varepsilon)P_W(\varepsilon).
\end{align*}
From (\ref{error probability}), (\ref{rr01}) and (\ref{d01}) we see
\begin{align*}
P_{\epsilon}(W)=\frac{1}{2}\sum_{\varepsilon\in[0,1]}\sum_{y\in L_W(\varepsilon)}\min\{\varepsilon,\overline{\varepsilon}\}\sum_{x\in\mathcal{X}}\Pr(y|x)
= \sum_{\varepsilon\in[0,1]}\min\{\varepsilon,\overline{\varepsilon}\}P_W(\varepsilon).
\end{align*}
From (\ref{Bp}), (\ref{rr01}) and (\ref{d01}) we see
\begin{align*}
Z(W)=\sum_{\varepsilon\in[0,1]}\sum_{y\in L_W(\varepsilon)}\sqrt{\varepsilon\overline{\varepsilon}}\sum_{x\in\mathcal{X}}\Pr(y|x)
=2\sum_{\varepsilon\in[0,1]}\sqrt{\varepsilon\overline{\varepsilon}}P_W(\varepsilon).
\end{align*}

The proof is completed.
\end{proof}

\subsection{Random Switching of BIDMCs}
For any positive integer $n$, let $[n]$ denote the set $\{0,1,\cdots\,n-1\}$ of integers.

Assume that $\{\mathcal{Y}_j\}_{j\in[n]}$ is a partition of the output alphabet $\mathcal{Y}$ of the BIDMC $W$ such that, for each $j\in[n]$, the probability
\begin{equation*}
	\Pr(y\in\mathcal{Y}_j| x=a)=q_j
\end{equation*}
is independent of $a\in\mathcal{X}$. Clearly, we have $q_j=\Pr(y\in\mathcal{Y}_j),\, j\in[n]$ and that $(q_0,q_1,\ldots,q_{n-1})$ is a probability distribution vector. For $j\in[n]$, let $W_j:x\in\mathcal{X}\mapsto y_j\in\mathcal{Y}_j$ denote the synthetic BIDMC with transition probabilities
\begin{equation}\label{60}
	\Pr(y_j| x)=\frac{1}{q_j}\Pr(y=y_j| x),
\end{equation}
which, in accordance with our assumption regarding the notations, should be understood as
\begin{equation*}
	\Pr(y_j=b| x=a)=\frac{1}{q_j}\Pr(y=b| x=a), \text{ for any }b\in\mathcal{Y}_j\text{ and }a\in\mathcal{X}.\label{60'}
\end{equation*}

Since the input $x$ of the channel $W$ may be seen as randomly being transmitted over the sub-channels $\{W_j\}_{j\in[n]}$, and over $W_j$ with probability $q_j$ for each $j\in[n]$,
as depicted in Figure~1,
$W$ is also called a \emph{random switching channel} (RSC) of $\{W_j\}_{j\in[n]}$, and written as
\begin{align}
W=\sum_{j\in[n]}q_jW_j.\label{rr02}
\end{align}
We note that each pair of the output alphabets of the sub-channels $\{W_j\}_{j\in[n]}$ are disjoint, it is supposed naturally that the receiver knows the exact sub-channel used for each transmission.

\setlength{\unitlength}{0.3cm}
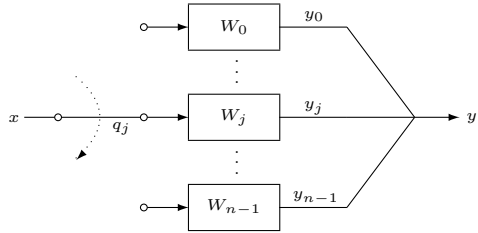
\begin{figure}[t]
\begin{center}
\begin{picture}(22,10)

\put(8,8){\framebox(4,2)[]{{\tiny $W_0$}}}
\put(8,4){\framebox(4,2)[]{{\tiny $W_j$}}}
\put(8,0){\framebox(4,2)[]{{\tiny $W_{n-1}$}}}
\put(10,2.5){\tiny $\vdots$}
\put(10,6.5){\tiny $\vdots$}
\multiput(6,1)(0,4){3}{\circle{.3}}
\multiput(6.15,1)(0,4){3}{\vector(1,0){1.85}}
\put(2.35,5){\line(1,0){3.5}}
\put(2.2,5){\circle{.3}}
\put(.7,5){\line(1,0){1.35}}
\qbezier[15](3.,3.2)(5.1,5)(3.,6.8)
\put(3.2,3.3){\vector(-3,-3){0.2}}
\put(0,4.8){\tiny $x$}
\multiput(12,1)(0,4){3}{\line(1,0){3}}
\put(18.,5){\line(-1,0){3}}
\put(18.,5){\line(-3,-4){3}}
\put(18.,5){\line(-3,4){3}}
\put(18.,5){\vector(1,0){2}}
\put(20.3,4.9){\tiny $y$}
\put(13.1,9.4){\tiny $y_0$}
\put(13.1,5.4){\tiny $y_j$}
\put(12.6,1.4){\tiny $y_{n-1}$}
\put(4.6,4.5){\tiny $q_j$}
\end{picture}
\end{center}
\caption{An RSC of BIDMCs $\{W_j\}_{j\in[n]}$. For each transmission, the sub-channel $W_j$ is chosen with probability $q_j$ to be the actually used channel.}
\end{figure}

For the RSCs of BIDMCs, we have the following theorem
which shows the LRP and the notation defined by (\ref{rr02}) admit some natural operations.

\begin{theorem}
Let $W=\sum_{j\in[n]}q_jW_j$ be an RSC of the BIDMCs $\{W_j\}_{j\in[n]}$.

\noindent
\begin{enumerate}
\item For any $\varepsilon\in[0,1]$, we have
\begin{align}
P_W(\varepsilon)=\sum_{j\in[n]}q_jP_{W_j}(\varepsilon).\label{62'}
\end{align}
Therefore,
\begin{gather}
I(W)=\sum_{j\in[n]}q_jI(W_j),
P_{\epsilon}(W)=\sum_{j\in[n]}q_jP_{\epsilon}(W_j),
Z(W)=\sum_{j\in[n]}q_jZ(W_j).\label{62}
\end{gather}
\item Assume $W'=\sum_{j\in[n]}q'_jW'_j$ is a BIDMC that is independent of\ $W$, where $W'_j$ is equivalent to $W_j$ for $j\in[n]$. Then,
for any $p\in[0,1]$ we have
\begin{align}
pW+\bar{p}W'\cong\sum_{j\in[n]}pq_jW_j+\sum_{j\in[n]}\overline{p}q'_jW'_j
\cong\sum_{j\in[n]}(pq_j+\overline{p}q'_j)W_j.\label{c00}
\end{align}
\end{enumerate}
\end{theorem}
\begin{proof}
Since $W:x\in\mathcal{X}\mapsto y\in\mathcal{Y}$ is an RSC of the BIDMCs $W_j:x\in\mathcal{X}\mapsto y_j\in\mathcal{Y}_j$,
for any $\varepsilon\in[0,1]$ the set $L_W(\varepsilon)$ has partition $\{L_{W_j}(\varepsilon)\}_{j\in[n]}$, and thus according to (\ref{rr01}) and (\ref{60}) we see
\begin{align*}
P_W(\varepsilon)=&\frac{1}{2}\sum_{x\in\mathcal{X},\,y\in L_W(\varepsilon)}\Pr(y|x)\\
=&\frac{1}{2}\sum_{x\in\mathcal{X},\,j\in[n],\,y_j\in L_{W_j}(\varepsilon)}q_j\Pr(y_j|x)\\
=&\sum_{j\in[n]}q_jP_{W_j}(\varepsilon),
\end{align*}
i.e., (\ref{62'}) is valid for any $\varepsilon\in[0,1]$.
Therefore,
we see further that (\ref{62}) and (\ref{c00}) follow from Theorem~\ref{lem00s} and the equivalence of BIDMCs,
respectively.
\end{proof}

For $\varepsilon\in[0,1]$, let $\mathrm{B}(\varepsilon)$ denote the BSC with crossover probability $ \varepsilon $ and $\mathrm{E}(\varepsilon)$ the BEC with erasure probability $\varepsilon$, respectively. Clearly, for any $\varepsilon,\sigma,q\in[0,1]$ we have $\mathrm{B}(\varepsilon)\cong\mathrm{B}(\overline{\varepsilon})$,
$q\mathrm{E}(\varepsilon)+\overline{q}\mathrm{E}(\sigma)\cong \mathrm{E}(q\varepsilon+\overline{q}\sigma)$ and
	$q\mathrm{B}(1/2)+\overline{q}\mathrm{B}(0)\cong \mathrm{E}(q)$, where $\mathrm{B}(0)$ and $\mathrm{B}(1/2)$ are the noiseless channel and the completely noisy channel respectively.

A BIDMC $W$ is said to be \emph{symmetric} if its LRP is symmetric with respect to $1/2$, i.e.,
\begin{gather}
P_W(\varepsilon)=P_W(\overline{\varepsilon}),\ \text{for }\varepsilon\in[0,1].
\end{gather}
Clearly, the BIDMC $W$ is symmetric if and only if it is equivalent to an RSC of some BSCs $\{\mathrm{B}(\varepsilon_i)\}_{i\in[n]}$
with $0\leq \varepsilon_0<\varepsilon_1<\cdots<\varepsilon_{n-1}\leq 1/2$.

Notice that the symmetry of BIDMCs defined here is based upon equivalence and slightly different from those defined in literature.
For example, for any $p\in(0,1)$ with $p\neq 1/2$ the channel given in Figure~2 is symmetric according to our definition, but according to that noted in \cite{Arikan09}.

\setlength{\unitlength}{0.3cm}
\begin{figure}[t]
\begin{center}
\begin{picture}(12,7)
\put(2,4){\vector(4,1){8}}
\put(2,4){\vector(8,-1){8}}
\put(2,4){\vector(4,-2){8}}
\put(2,2){\vector(4,2){8}}
\put(2,2){\vector(8,1){8}}
\put(2,2){\vector(4,-1){8}}
\put(1,3.8){\mbox{\tiny $0$}}
\put(1,1.8){\mbox{\tiny $1$}}
\put(10.5,5.8){\mbox{\tiny $-1$}}
\put(10.5,2.8){\mbox{\tiny $0$}}
\put(10.5,-0.2){\mbox{\tiny $+1$}}
\put(5.5,.2){\mbox{\tiny $\overline{p}/3$}}
\put(5.5,5.5){\mbox{\tiny $2p/3$}}
\put(8.3,1.2){\mbox{\tiny $p/3$}}
\put(8.3,4.5){\mbox{\tiny $2\overline{p}/3$}}
\put(7.2,2.2){\mbox{\tiny $p$}}
\put(7.2,3.7){\mbox{\tiny $\overline{p}$}}

\end{picture}
\end{center}
\caption{A symmetric BIDMC which is equivalent to $\mathrm{B}(p)$.}
\end{figure}
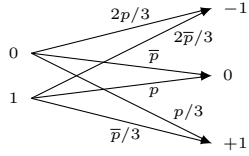

\subsection{Degradation of BIDMCs}
As in \cite{Vardy13}, for any BIDMC $W:x\in \mathcal{X}\mapsto y\in\mathcal{Y}$, another
BIDMC $W':x\in \mathcal{X}\mapsto y'\in\mathcal{Y}'$ is said a \emph{degadation channel} of $W$, and written as $W'\preccurlyeq W$,
if there is a channel $Q:\mathcal{Y}\rightarrow \mathcal{Y}'$ such that
\begin{align}
\Pr(y'|x)=\sum_{y\in\mathcal{Y}}\Pr(y|x)Q(y'|y),\label{c05}
\end{align}
where $Q(b|a)$ is the probability of that $Q$ transmits $a\in\mathcal{Y}$ into $b\in\mathcal{Y}'$. The channel $Q$ is also called the \emph{intermediate channel} of the degradation $W'\preccurlyeq W$.
Clearly, from $W'\preccurlyeq W$ and $W\preccurlyeq W''$ one can get $W'\preccurlyeq W''$, i.e., the degradation of BIDMCs has {\it transitivity}.
The following lemma is useful for the argument on the equivalence of BIDMCs.

\begin{lemma}\label{lem000}
Assume that $W:x\in\mathcal{X}\mapsto y\in\mathcal{Y}$ and $W':x\in\mathcal{X}\mapsto y'\in\mathcal{Y}'$ are two BIDMCs with $W'\preccurlyeq W$.
Then, we have
\begin{gather}
\min\{\varepsilon\in[0,1]:P_W(\varepsilon)\neq 0\}\leq
\min\{\varepsilon\in[0,1]:P_{W'}(\varepsilon)\neq 0\},\label{c20}\\
\max\{\varepsilon\in[0,1]:P_W(\varepsilon)\neq 0\}\geq
\max\{\varepsilon\in[0,1]:P_{W'}(\varepsilon)\neq 0\}.\label{c21}
\end{gather}
\end{lemma}
\begin{proof}
Let $\varepsilon_0$ denote the left side and $\varepsilon'_0$ the right side, respectively, of (\ref{c20}).
Assume that $y'$ is an arbitrary element in $L_{W'}(\varepsilon'_0)$ and $Q$ is the intermediate channel of the degradation $W' \preccurlyeq W$. Since $Q$ transmits $y\in\mathcal{Y}$ into $y'\in\mathcal{Y}'$ with probability $Q(y'|y)$ that satisfies (\ref{c05}), we have
\begin{align}
\frac{\overline{\varepsilon'_0}}{\varepsilon'_0}=&\frac{\Pr(y'|x=0)}{\Pr(y'|x=1)}
=\frac{1}{\Pr(y'|x=1)}\sum_{\varepsilon\in [0,1]}\sum_{y\in L_W(\varepsilon)}\Pr(y|x=0)Q(y'|y)\nonumber\\
=&\frac{1}{\Pr(y'|x=1)}\sum_{\varepsilon\in [0,1]}\frac{\overline{\varepsilon}}{\varepsilon}\sum_{y\in L_W(\varepsilon)}\Pr(y|x=1)Q(y'|y)\label{c06}\\
\leq& \frac{\overline{\varepsilon_0}}{\varepsilon_0\Pr(y'|x=1)}\sum_{y\in\mathcal{Y}}\Pr(y|x=1)Q(y'|y)
=\frac{\overline{\varepsilon_0}}{\varepsilon_0},\nonumber
\end{align}
which implies (\ref{c20}). The inequality (\ref{c21}) can be proved similarly.
\end{proof}

A BIDMC W is referred to as being {\it LRP-oriented} provided that the set $L_{W}(\varepsilon)$ contains at most one element for every $\varepsilon\in[0,1]$.
Evidently, for any BIDMC $W$, there exists a unique BIDMC $W'$, which is named the LRP-oriented form of $W$, such that $W'$ is LRP-oriented and equivalent to $W$.
The following theorem demonstrates that the equivalence of BIDMCs defined in this paper is truly equivalent to that defined in \cite{Vardy13}.

\begin{theorem}\label{theo20}
For any BIDMCs $W:x\in\mathcal{X}\mapsto y\in\mathcal{Y}$ and $W':x\in\mathcal{X}\mapsto y'\in\mathcal{Y}'$,
we have $W\preccurlyeq W'\preccurlyeq W$ if and only if $W\cong W'$, i.e., $P_{W}(\varepsilon)$ equals $P_{W'}(\varepsilon)$ for any $\varepsilon\in[0,1]$.
\end{theorem}
\begin{proof}
Assume $W\cong W'$. Let $Q'$ denote the channel which transmits $b\in\mathcal{Y}$ into
$b'\in\mathcal{Y}'$ with probability
$Q'(b'|b)$ given by
\begin{align*}
\frac{\Pr(y'=b'|x=0)+\Pr(y'=b'|x=1)}{2P_W(\varepsilon)}
\end{align*}
if $b\in L_W(\varepsilon)$ and $b'\in L_{W'}(\varepsilon)$
for some $\varepsilon\in[0,1]$, and by 0 otherwise.
Since for any $b\in L_W(\varepsilon)$ from (\ref{rr01}) we see
\begin{align*}
&\sum_{b'\in\mathcal{Y}'}Q'(b'|b)=\sum_{b'\in L_{W'}(\varepsilon)}Q'(b'|b)\\
=&\frac{1}{2P_W(\varepsilon)}\sum_{b'\in L_{W'}(\varepsilon)}(\Pr(y'=b'|x=0)+\Pr(y'=b'|x=1))
=\frac{2P_{W'}(\varepsilon)}{2P_W(\varepsilon)}=1,
\end{align*}
$Q'$ is a well-defined channel.
Therefore, for any $\varepsilon\in[0,1]$, $a\in\mathcal{X}$ and $b\in L_{W'}(\varepsilon)$ from (\ref{d01}) we have
\begin{align*}
&\sum_{y\in\mathcal{Y}}\Pr(y|x=a)Q'(y'=b|y)\\
=&\sum_{y\in L_W(\varepsilon)}\frac{\Pr(y|x=a)(\Pr(y'=b|x=0)+\Pr(y'=b|x=1))}{2P_W(\varepsilon)}\\
=&\sum_{y\in L_W(\varepsilon)}\frac{\Pr(y'=b|x=a)(\Pr(y|x=0)+\Pr(y|x=1))}{2P_W(\varepsilon)}
=\Pr(y'=b|x=a),
\end{align*}
i.e., (\ref{c05}) is valid for $Q=Q'$. Hence, we have $W'\preccurlyeq W$. Similarly, we can also get $W\preccurlyeq W'$. Therefore, we have proved the if-part of this theorem.

To show the only-if-part of this theorem, without loss of generality we assume that $W$ and $W'$ are LRP-oriented BIDMCs with
$W\preccurlyeq W'\preccurlyeq W$.
Let $y_0,y_1,\ldots,y_{n-1}$ be the elements in $\mathcal{Y}$ ordered according to $\varepsilon_0<\varepsilon_1<\cdots<\varepsilon_{n-1}$, where $\varepsilon_i\in [0,1]$ is the number satisfying
$y_i\in L_W(\varepsilon_i)$ for each $i\in[n]$.
Let $y'_0,y'_1,\ldots,y'_{m-1}$ be the elements in $\mathcal{Y}'$ ordered according to $\varepsilon'_0<\varepsilon'_1<\cdots<\varepsilon'_{m-1}$, where $\varepsilon'_i\in [0,1]$ is the number satisfying
$y'_i\in L_{W'}(\varepsilon'_i)$ for each $i\in[m]$. Assume that $Q$ is the intermediate channel of the degradation $W' \preccurlyeq W$, then $Q$ transmits $y\in\mathcal{Y}$ into $y'\in\mathcal{Y}'$ with probability $Q(y'|y)$ that satisfies (\ref{c05}).

According to $W\preccurlyeq W'\preccurlyeq W$ and Lemma~\ref{lem000}, we see $\varepsilon_0=\varepsilon'_0$.
Therefore, from (\ref{c06}) we have
$Q(y'_0|y_i)=0$ for $i>0$,
and then from (\ref{c05}) we see
\begin{align*}\Pr(y'_0|x=0)=\Pr(y_0|x=0)Q(y'_0|y_0)\leq \Pr(y_0|x=0).
\end{align*}
Similarly, one can also get $\Pr(y_0|x=0)\leq \Pr(y'_0|x=0)$.
Hence, we have
$P_W(\varepsilon_0)=P_{W'}(\varepsilon_0)$, $Q(y'_0|y_0)=1$
and $Q(y'_i|y_0)=0$ for $i>0$.

Furthermore, according to similar arguments one can show by induction, for any $j\geq 0$,
\begin{gather*}
\varepsilon_j=\varepsilon'_j,\ P_W(\varepsilon_j)=P_{W'}(\varepsilon_j),\ Q(y'_j|y_j)=1,\\
Q(y'_j|y_i)=Q(y'_i|y_j)=0\text{ for }i>j.
\end{gather*}
Hence, we must have $W\cong W'$.
\end{proof}

The following lemma encompasses several well-known results regarding degradation channels (see, for example, \cite{Urbanke08}, \cite{Vardy13}, \cite{Urbanke19}, etc.). For the sake of completeness, we present a new proof for it here.

\begin{lemma}\label{lem000'}
Assume that $W:x\in\mathcal{X}\mapsto y\in\mathcal{Y}$ and $W':x\in\mathcal{X}\mapsto y'\in\mathcal{Y}'$ are two BIDMCs with $W'\preccurlyeq W$.
Then, we have
\begin{gather}
I(W)\geq I(W'),\
P_{\epsilon}(W)\leq P_{\epsilon}(W'),\
Z(W)\leq Z(W').\label{c26}
\end{gather}
\end{lemma}
\begin{proof}
Without loss of generality, we assume that $W$, $W'$ are LRP-oriented, $\mathcal{Y}=\{b_i:i\in[n]\}$ and $\mathcal{Y}'=\{b'_j:j\in[m]\}$.
For $i\in[n]$ and $j\in[m]$, let $\varepsilon_i$ and $\varepsilon'_j$ be the numbers in $[0,1]$ with
$L_W(\varepsilon_i)=\{b_i\}$ and $L_{W'}(\varepsilon'_j)=\{b'_j\}$, respectively.  Assume that the intermediate channel of $W' \preccurlyeq W$ transmits $b_i$ into $b'_j$ with probability $q_{i,j}$.
Then, for $j\in[m]$ we have
\begin{gather*}
P_{W'}(\varepsilon'_j)=\frac{1}{2}\sum_{x\in\mathcal{X}}\Pr(y'=b'_j|x)
=\frac{1}{2}\sum_{x\in\mathcal{X}}\sum_{i\in[n]}\Pr(y=b_i|x)q_{i,j}
=\sum_{i\in[n]}P_W(\varepsilon_i)q_{i,j},\\
\varepsilon'_j=\frac{\Pr(y'=b'_j|x=1)}{2P_{W'}(\varepsilon'_j)}
=\frac{\sum_{i\in[n]}\Pr(y=b_i|x=1)q_{i,j}}{2P_{W'}(\varepsilon'_j)}
=\sum_{i\in[n]}\varepsilon_i \frac{P_W(\varepsilon_i)q_{i,j}}{P_{W'}(\varepsilon'_j)}.
\end{gather*}
Therefore, according to Theorem~\ref{lem00s} and the convexity of the functions $\hbar(\varepsilon)$, $\min\{\varepsilon,\overline{\varepsilon}\}$ and $\sqrt{\varepsilon\overline{\varepsilon}}$,
by using the well-known Jensen inequality we get
\begin{gather*}
I(W)=1-\sum_{j\in[m]}\sum_{i\in[n]}\hbar(\varepsilon_i)P_W(\varepsilon_i)q_{i,j}\geq 1-\sum_{j\in[m]}\hbar(\varepsilon'_j)P_{W'}(\varepsilon'_j)=I(W'),\\
P_{\epsilon}(W)=
\sum_{j\in[m]}\sum_{i\in[n]}\min\{\varepsilon_i,\overline{\varepsilon_i}\}
P_W(\varepsilon_i)q_{i,j}\leq \sum_{j\in[m]}\min\{\varepsilon'_j,\overline{\varepsilon'_j}\}
P_{W'}(\varepsilon'_j)=P_{\epsilon}(W'),\\
Z(W)=2\sum_{j\in[m]}\sum_{i\in[n]}\sqrt{\varepsilon_i\overline{\varepsilon_i}}P_W(\varepsilon_i)q_{i,j}
\leq 2\sum_{j\in[m]}\sqrt{\varepsilon'_j\overline{\varepsilon'_j}}P_{W'}(\varepsilon'_j)
=Z(W').
\end{gather*}

The proof is completed.
\end{proof}

\section{Arikan Transformations of BIDMCs}
\label{sec03}

Let $(u_0,u_1)$ be a random vector uniformly distributed over $\mathcal{X}^2$.
Assume that $u_0+u_1$ is transmitted over $W_0:x_0\in\mathcal{X}\mapsto y_0\in\mathcal{Y}_0$, whereas $u_1$ is transmitted over $W_1:x_1\in\mathcal{X}\mapsto y_1\in\mathcal{Y}_1$, respectively,
where the BIDMCs $W_0$ and $W_1$ are independent.
Let $A_0(W_0,W_1)$ denote the synthetic BIDMC: $u_0\mapsto(y_0,y_1)$ with transition probabilities
\begin{align}
  \Pr(y_0,y_1|u_0)=\frac{1}{2}\sum_{u_1\in\mathcal{X}}\Pr(y_0|x_0=u_0+u_1)\Pr(y_1|x_1=u_1).\label{at06}
\end{align}
Let $A_1(W_0,W_1)$ denote the synthetic BIDMC: $u_1\mapsto(y_0,y_1,u_0)$ with transition probabilities
\begin{align}
  \Pr(y_0,y_1,u_0|u_1)=\frac{1}{2}\Pr(y_0|x_0=u_0+u_1)\Pr(y_1|x_1=u_1).\label{at07}
\end{align}

We note that $A_0(W_0,W_1)$ and $A_1(W_0,W_1)$ are just the synthetic BIDMCs denoted by $W_0\boxtimes W_1$ and $W_0\circledast W_1$, respectively, in \cite{Hirche18}.
Since these synthetic BIDMCs played important roles in the proposal of Arikan's polar codes, we will call them and their compositions as \emph{Arikan Transformations} of some underlying BIDMCs in this paper.
For example,
$A_0(A_1(W_0,W_1),W_2)$ is an Arikan transformation of independent BIDMCs $W_0,W_1,W_2$.

\subsection{Some Elementary Properties of Arikan Transformations}

In this subsection we show some elementary properties of Arikan transformations of BIDMCs.

Firstly, we prove that any Arikan transformation will degrade when the underlying channels are substituted with some channels that are degradations of the original ones.
Actually, we have the following theorem, which is a generalization of Lemma 5 in \cite{Vardy13}.
\begin{theorem}
\label{theo21}
If $W'_0:x_0\in\mathcal{X}\mapsto y'_0\in\mathcal{Y}'_0$ and $W'_1:x_1\in\mathcal{X}\mapsto y'_1\in\mathcal{Y}'_1$ are BIDMCs with $W'_0\preccurlyeq W_0$ and $W'_1\preccurlyeq W_1$, then
we have
\begin{align}\label{at03}
A_0(W'_0,W'_1)\preccurlyeq A_0(W_0,W_1),\\
A_1(W'_0,W'_1)\preccurlyeq A_1(W_0,W_1).\label{at04}
\end{align}
\end{theorem}
\begin{proof}
Let $Q_i$ be the intermediate channel of the degradation $W'_i\preccurlyeq W_i$ for $i=0,1$. According to (\ref{c05}) and (\ref{at06}), we have
\begin{align}
&\Pr(y'_0,y'_1|u_0)=\frac{1}{2}\sum_{a\in\mathcal{X}}\Pr(y'_0|x_0=u_0+a)\Pr(y'_1|x_1=a)\nonumber\\
=&\frac{1}{2}\sum_{a\in\mathcal{X}}
\sum_{y_0\in\mathcal{Y}_0}\Pr(y_0|x_0=u_0+a)Q_0(y'_0|y_0)
\sum_{y_1\in\mathcal{Y}_1}\Pr(y_1|x_1=a)Q_1(y'_1|y_1)\nonumber\\
=&\sum_{y_0\in\mathcal{Y}_0,y_1\in\mathcal{Y}_1}\Pr(y_0,y_1|u_0)Q_0(y'_0|y_0)Q_1(y'_1|y_1)\nonumber\\
=&\sum_{(y_0,y_1)\in\mathcal{Y}_0\times\mathcal{Y}_1}\Pr(y_0,y_1|u_0)T(y'_0,y'_1|y_0,y_1),\label{at07}
\end{align}
where $T(y'_0,y'_1|y_0,y_1)=Q_0(y'_0|y_0)Q_1(y'_1|y_1)$
is the transition probability of some channel
$T: (y_0,y_1)\in\mathcal{Y}_0\times\mathcal{Y}_1\mapsto (y'_0,y'_1)\in\mathcal{Y}'_0\times\mathcal{Y}'_1$ since we have
\begin{gather*}
\sum_{(y'_0,y'_1)\in\mathcal{Y}'_0\times\mathcal{Y}'_1}T(y'_0,y'_1|y_0,y_1)=
\sum_{(y'_0,y'_1)\in\mathcal{Y}'_0\times\mathcal{Y}'_1}Q_0(y'_0|y_0)Q_1(y'_1|y_1)=1.
\end{gather*}
Hence, from (\ref{at07}) we see (\ref{at03}).
According to (\ref{c05}) and (\ref{at07}), we have
\begin{align}
&\Pr(y'_0,y'_1,u_0|u_1)=\frac{1}{2}\Pr(y'_0|x_0=u_0+u_1)\Pr(y'_1|x_1=u_1)\nonumber\\
=&\frac{1}{2}\sum_{y_0\in\mathcal{Y}_0}\Pr(y_0|x_0=u_0+u_1)Q_0(y'_0|y_0)
\sum_{y_1\in\mathcal{Y}_1}\Pr(y_1|x_1=u_1)Q_1(y'_1|y_1)\nonumber\\
=&\sum_{y_0\in\mathcal{Y}_0,y_1\in\mathcal{Y}_1}\Pr(y_0,y_1,u_0|u_1)Q_0(y'_0|y_0)Q_1(y'_1|y_1)
\nonumber\\
=&\sum_{(y_0,y_1,u)\in\mathcal{Y}_0\times\mathcal{Y}_1\times\mathcal{X}}
\Pr(y_0,y_1,u|u_1)R(y'_0,y'_1,u_0|y_0,y_1,u),
\label{at05}
\end{align}
where $R(y'_0,y'_1,u_0|y_0,y_1,u)$ equals to $Q_0(y'_0|y_0)Q_1(y'_1|y_1)$ if $u=u_0$, and 0 otherwise.
From
\begin{align*}
\sum_{(y'_0,y'_1,u_0)\in\mathcal{Y}'_0\times\mathcal{Y}'_1
\times\mathcal{X}}R(y'_0,y'_1,u_0|y_0,y_1,u)
=\sum_{(y'_0,y'_1)\in\mathcal{Y}'_0\times\mathcal{Y}'_1}Q_0(y'_0|y_0)Q_1(y'_1|y_1)=1
\end{align*}
we see that $R(y'_0,y'_1,u_0|y_0,y_1,u)$ is the transition probability of some channel
$R: (y_0,y_1,u)\in\mathcal{Y}_0\times\mathcal{Y}_1\times\mathcal{X}\mapsto (y'_0,y'_1,u_0)\in\mathcal{Y}'_0\times\mathcal{Y}'_1\times\mathcal{X}$.
Therefore, from (\ref{at05}) we see (\ref{at04}).
\end{proof}

Based on Theorems~\ref{theo20} and \ref{theo21}, the following corollary can be easily deduced. This corollary demonstrates that any Arikan transformation will remain invariant, in terms of equivalence, whenever any of the underlying channels is substituted by some of its equivalent channels.
\begin{corollary}
If $W'_0\cong W_0$ and $W'_1\cong W_1$, then we have
\begin{gather}
A_0(W'_0,W'_1)\cong A_0(W_0,W_1),\
A_1(W'_0,W'_1)\cong A_1(W_0,W_1).\label{at09}
\end{gather}
\end{corollary}

It should be noted that this corollary actually constitutes a generalization of Proposition 4 in \cite{Ye18}. The reason for this is that our definition of the equivalence of BIDMCs is somewhat less restrictive compared to the one provided therein.

Furthermore, as shown in the following theorem, Arikan transformations also admit a natural operation for the RSCs of BIDMCs.

\begin{theorem}\label{lem200}
If $W_0=\sum_{i\in[n]}p_iW'_i$ and $W_1=\sum_{j\in[m]}q_jW''_j$, where $W'_i:x_0\in\mathcal{X}\mapsto
y'_i\in\mathcal{Y}'_i$ and $W''_j:x_1\in\mathcal{X}\mapsto
y''_j\in\mathcal{Y}''_j$ are independent BIDMCs for $i\in[n']$ and $j\in[m]$, then
\begin{align}\label{at00}
A_0(W_0,W_1)\cong\sum_{i\in[n],j\in[m]}p_iq_jA_0(W'_i,W''_j),\\
A_1(W_0,W_1)\cong\sum_{i\in[n],j\in[m]}p_iq_jA_1(W'_i,W''_j).\label{at01}
\end{align}
\end{theorem}
\begin{proof}
At first, we note that $\mathcal{Y}_0=\cup_{i\in[n]}\mathcal{Y}'_i$ and $\mathcal{Y}_1=\cup_{j\in[m]}\mathcal{Y}''_j$, the output alphabets of $W_0$ and $W_1$, have
partitions $\{\mathcal{Y}'_i\}_{i\in[n]}$ and $\{\mathcal{Y}''_j\}_{j\in[m]}$ respectively.

For $a_i\in\mathcal{Y}'_i$ and $b_j\in\mathcal{Y}''_j$, we have
\begin{align}
&\Pr(y_0=a_i,y_1=b_j|u_0)\nonumber\\
=&\frac{1}{2}\sum_{u_1\in\mathcal{X}}\Pr(y_0=a_i|x_0=u_0+u_1)\Pr(y_1=b_j|x_1=u_1)\nonumber\\
=&\frac{1}{2}\sum_{u_1\in\mathcal{X}}p_i\Pr(y'_i=a_i|x_0=u_0+u_1)q_j\Pr(y''_j=b_j|x_1=u_1)\nonumber\\
=&p_iq_j\Pr(y'_i=a_i,y''_j=b_j|u_0),\label{yy0}
\end{align}
and
\begin{align}
&\Pr(y_0=a_i,y_1=b_j,u_0|u_1)\nonumber\\
=&\frac{1}{2}\Pr(y_0=a_i|x_0=u_0+u_1)\Pr(y_1=b_j|x_1=u_1)\nonumber\\
=&\frac{1}{2}p_i\Pr(y'_i=a_i|x_0=u_0+u_1)q_j\Pr(y''_j=b_j|x_1=u_1)\nonumber\\
=&p_iq_j\Pr(y'_i=a_i,y''_j=b_j,u_0|u_1).\label{yy1}
\end{align}
Hence, we see
\begin{align*}
  &\sum_{(a_i,b_j)\in\mathcal{Y}'_i\times\mathcal{Y}''_j}\Pr(y_0=a_i,y_1=b_j|u_0)\\
  =&\sum_{(a_i,b_j)\in\mathcal{Y}'_i\times
  \mathcal{Y}''_j}p_iq_j\Pr(y'_i=a_i,y''_j=b_j|u_0)=p_iq_j
\end{align*}
is independent of $u_0$, and
\begin{align*}
  &\sum_{(a_i,b_j,u_0)\in\mathcal{Y}'_i\times\mathcal{Y}''_j\times
  \mathcal{X}}\Pr(y_0=a_i,y_1=b_j,u_0|u_1)\\
  =&\sum_{(a_i,b_j,u_0)\in\mathcal{Y}'_i\times\mathcal{Y}''_j\times
  \mathcal{X}}p_iq_j\Pr(y'_i=a_i,y''_j=b_j,u_0|u_1)=p_iq_j
\end{align*}
is also independent of $u_1$. Therefore, according to the definitions of the transformations $A_0(W_0,W_1)$, $A_1(W_0,W_1)$ and the random switching of BIDMCs, we see that (\ref{at00}) and (\ref{at01}) follow from (\ref{yy0}) and (\ref{yy1}) respectively.
\end{proof}
\subsection{Arikan Transformations of Symmetric BIDMCs}

In this subsection, we deal with the Arikan transformations when the underlying channels are symmetric BIDMCs.

For real numbers $a,b\in[0,1]$, let $a\star b=\bar{a}b+a\bar{b}$ and
\begin{gather}a\diamond b=\left\{\begin{array}{ll}
ab/(\overline{a}\star b),&\text{if }\{a,b\}\subset(0,1),\\
0,&\text{if }\{a,b\}\cap \{0,1\}\neq \emptyset.
\end{array}\right.\end{gather}
One can deduce the following lemma easily.
\begin{lemma}\label{lem201}
For the operations $\star$ and $\diamond$, we have
\begin{gather}
a\star b=b\star a,\ (a\star b)\star c=a \star(b\star c),\label{gg0}\\
a\diamond b=b\diamond a,\ (a\diamond b)\diamond c=a \diamond(b\diamond c),\label{gg1}\\
\overline{a\star b}=\overline{a}\star b,\ a\star \frac{1}{2}=\frac{1}{2},\label{gg2}\\
 \overline{a\diamond b}=\overline{a}\diamond \overline{b}\text{ for }\{a,b\}\subset (0,1),\label{gg3}\\
a\diamond \overline{a}=\frac{1}{2}\text{ and }a\diamond \frac{1}{2}=a\ \text{for }a\in(0,1).\label{gg4}
\end{gather}
Furthermore, $\langle(0,1),\diamond\rangle$ is an Abel group and
$$
\langle(0,1),\star\rangle,\ \langle[0,1/2],\star\rangle,\ \langle[0,1/2],\diamond\rangle
$$
are commutative semi-groups.
\end{lemma}

When the underlying BIDMCs are BSCs, the Arikan transformations assume simple forms as presented in the following lemma.
\begin{lemma}\label{lem202}
For any $\varepsilon_0,\varepsilon_1\in[0,1]$, we have
\begin{gather}
A_0(\mathrm{B}(\varepsilon_0),\mathrm{B}(\varepsilon_1))
\cong\mathrm{B}({\varepsilon_0}\star\varepsilon_1)\label{at31},\\
A_1(\mathrm{B}(\varepsilon_0),\mathrm{B}(\varepsilon_1))\cong
({\varepsilon_0}\star\overline{\varepsilon_1})
\mathrm{B}({\varepsilon_0}\diamond{\varepsilon_1})
+({\varepsilon_0}\star{\varepsilon_1})
\mathrm{B}({\varepsilon_0}\diamond\overline{\varepsilon_1}).\label{at32}
\end{gather}
\end{lemma}
\begin{proof}
Since the channel $A_0(\mathrm{B}(\varepsilon_0),\mathrm{B}(\varepsilon_1))$ has transition probabilities
\begin{align*}
  \Pr(y_0,y_1|u_0)
  =&\frac{1}{2}\sum_{u_1\in\mathcal{X}}\Pr(y_0|x_0=u_0+u_1)\Pr(y_1|x_1=u_1)\\
  =&\left\{\begin{array}{ll}
  \frac{1}{2}(\varepsilon_0\varepsilon_1+\overline{\varepsilon_0}\, \overline{\varepsilon_1}), &\text{if }y_0+y_1=u_0,\\
  \frac{1}{2}(\overline{\varepsilon_0}\varepsilon_1+\varepsilon_0\overline{\varepsilon_1}), &\text{if }y_0+y_1=u_0+1,
  \end{array}
  \right.
\end{align*}
we see (\ref{at31}) follows. Since the channel $A_1(\mathrm{B}(\varepsilon_0),\mathrm{B}(\varepsilon_1))$ has transition probabilities
\begin{align*}
  \Pr(y_0,y_1,u_0|u_1)
  =&\frac{1}{2}\Pr(y_0|x_0=u_0+u_1)\Pr(y_1|x_1=u_1)\\
  =&\left\{\begin{array}{ll}
  \frac{1}{2}\varepsilon_0\varepsilon_1, &\text{if }y_0+u_0=y_1=u_1+1,\\
  \frac{1}{2}\overline{\varepsilon_0}\, \overline{\varepsilon_1}, &\text{if }y_0+u_0=y_1=u_1,\\
  \frac{1}{2}\overline{\varepsilon_0}\varepsilon_1, &\text{if }y_0+u_0=y_1+1=u_1,\\
  \frac{1}{2}\varepsilon_0\overline{\varepsilon_1}, &\text{if }y_0+u_0=y_1+1=u_1+1,
  \end{array}
  \right.
\end{align*}
we see (\ref{at32}) follows.
\end{proof}

Furthermore, for the Arikan transformations whose underlying BIDMCs are symmetric, we have the following lemma.
\begin{lemma}\label{lem301}
For symmetric BIDMCs $W$, $W'$ and $W''$, we have
\begin{gather}
  A_0(\mathrm{B}(1/2),W)\cong \mathrm{B}(1/2),\label{mm300} \\
  A_1(\mathrm{B}(1/2),W)\cong
  A_0(\mathrm{B}(0),W)\cong W, \label{mm301}\\
  A_1(\mathrm{B}(0),W)\cong \mathrm{B}(0),\label{mm302}\\
A_0(W,W')\cong A_0(W',W),\label{mm303}\\
A_1(W,W')\cong A_1(W',W),\label{mm304}\\
A_0(A_0(W,W'),W'')\cong A_0(W,A_0(W',W'')),\label{mm305}\\
A_1(A_1(W,W'),W'')\cong A_1(W,A_1(W',W'')).\label{mm306}
\end{gather}
In particular, the BIDMCs $A_0(W,W')$ and $A_1(W,W')$ are symmetric.
\end{lemma}
\begin{proof}
From Lemmas~\ref{lem201} and \ref{lem202}, we see easily that (\ref{mm300}) to (\ref{mm305}) are valid when $W,W',W''$ are BSCs.
According to (\ref{mm302}) we see that (\ref{mm306}) is valid when at least one channel in $\{W,W',W''\}$
is $\mathrm{B}(0)$.
Furthermore, for $\varepsilon,\sigma,\delta\in(0,1)$, from $({\varepsilon}\star\overline{\sigma})({\varepsilon}\diamond{\sigma})
={\varepsilon}{\sigma}$, ${\varepsilon}\diamond\overline{\sigma}\diamond\overline{\delta}=
\overline{\overline{\varepsilon}\diamond{\sigma}\diamond\delta}$, (\ref{at01}) and (\ref{at32}) we have
\begin{align*}
&A_1(A_1(\mathrm{B}(\varepsilon),\mathrm{B}(\sigma)),\mathrm{B}(\delta))\\
\cong&A_1\big(({\varepsilon}\star\overline{\sigma})
\mathrm{B}({\varepsilon}\diamond{\sigma})
+({\varepsilon}\star{\sigma})
\mathrm{B}({\varepsilon}\diamond\overline{\sigma}),\mathrm{B}(\delta)\big)\\
\cong&({\varepsilon}\star\overline{\sigma})\big((({\varepsilon}\diamond{\sigma})\star\overline{\delta})
\mathrm{B}(({\varepsilon}\diamond{\sigma})\diamond{\delta})
+(({\varepsilon}\diamond{\sigma})\star{\delta})
\mathrm{B}((\varepsilon\diamond{\sigma})\diamond\overline{\delta})\big)+\\
&\ \ \ ({\varepsilon}\star{\sigma})\big((({{\varepsilon}}\diamond\overline{\sigma})\star\overline{\delta})
\mathrm{B}(({{\varepsilon}}\diamond\overline{\sigma})\diamond{\delta})
+(({{\varepsilon}}\diamond\overline{\sigma})\star{\delta})
\mathrm{B}(({\varepsilon}\diamond\overline{\sigma})\diamond\overline{\delta})\big)\\
\cong&(\varepsilon\sigma\delta+\overline{\varepsilon}\,\overline{\sigma}\overline{\delta})
\mathrm{B}(\varepsilon\diamond\sigma\diamond\delta)+
(\varepsilon\sigma\overline{\delta}+\overline{\varepsilon}\,\overline{\sigma}{\delta})
\mathrm{B}(\varepsilon\diamond\sigma\diamond\overline{\delta})+\\
&\ \ \ ({\varepsilon}\overline{\sigma}\delta+\overline{\varepsilon}{\sigma}\overline{\delta})
\mathrm{B}({\varepsilon}\diamond\overline{\sigma}\diamond\delta)+
(\varepsilon\overline{\sigma}\overline{\delta}+\overline{\varepsilon}{\sigma}{\delta})
\mathrm{B}(\overline{\varepsilon}\diamond{\sigma}\diamond\delta)
\end{align*}
and thus we see that (\ref{mm306}) is valid when $W,W',W''$ are BSCs.
Hence, from Theorem \ref{lem200} we see that (\ref{mm300}) to (\ref{mm306}) are valid for any symmetric BIDMCs $W,W',W''$.
Clearly, the BIDMCs $A_0(W,W')$ and $A_1(W,W')$ are symmetric since they are equivalent to some RSCs of BSCs.
\end{proof}

For independent and symmetric BIDMCs $W_0,W_1,\ldots$, we define further $\Delta(W_0)=\nabla(W_0)=W_0$
and, for $t\geq 1$,
\begin{gather}
\Delta(\{W_i\}_{i\in[t+1]})=A_0(\Delta(\{W_i\}_{i\in[t]}),W_{t}),\label{u01}\\
\nabla(\{W_i\}_{i\in[t+1]})=A_1(\nabla(\{W_i\}_{i\in[t]}),W_{t}).\label{u02}
\end{gather}
Clearly, according to Lemma~\ref{lem301},
for $1\leq j<t$ we have
\begin{gather}
\Delta(\{W_i\}_{i\in[t]})\cong A_0(\Delta(\{W_i\}_{i\in[j]}),\Delta(\{W_{i+j}\}_{i\in[t-j]})),\label{u05}\\
\nabla(\{W_i\}_{i\in[t]})\cong A_1(\nabla(\{W_i\}_{i\in[j]}),\nabla(\{W_{i+j}\}_{i\in[t-j]}),\label{u06}
\end{gather}
and for any permutation $\{W'_i\}_{i\in[t]}$ of the channels $\{W_i\}_{i\in[t]}$ we have
\begin{gather}
\Delta(\{W'_i\}_{i\in[t]})\cong \Delta(\{W_i\}_{i\in[t]}),\label{u03}\\
\nabla(\{W'_i\}_{i\in[t]})\cong \nabla(\{W_i\}_{i\in[t]}).\label{u04}
\end{gather}

The vector $(x_0,x_1,\ldots,x_{n-1})$ is also denoted by $x_0^n$ if no confusion.
For any vector $\sigma_0^t\in[0,1]^t$, let
\begin{align}
\varpi(\sigma_0^t)=\left\{
\begin{array}{ll}
\frac{1}{2}\big(\prod_{i\in [t]}\sigma_i+\prod_{i\in[t]}\overline{\sigma_i}\big),
& \text{if }\sigma_0^t\in(0,1)^t,
\\
2^{-t},
& \text{if }\sigma_0^t\not\in(0,1)^t.
\end{array}\right.\label{u07}
\end{align}
For $\sigma_0^{t+1}\in(0,1)^{t+1}$, from the definition of the operation $\diamond$, by induction one can prove easily
\begin{gather*}
\diamond_{i\in[t]}\sigma_{i}=\frac{\prod_{i\in[t]}\sigma_{i}}
{\prod_{i\in[t]}\sigma_{i}+\prod_{i\in[t]}\overline{\sigma_{i}}}
=\frac{\prod_{i\in[t]}\sigma_{i}}{2\varpi(\sigma_0^t)},
\end{gather*}
and then from $\overline{\diamond_{i\in[t]}\sigma_{i}}=\diamond_{i\in[t]}\overline{\sigma_{i}}$ we see
\begin{gather}
\varpi(\sigma_0^t)\big((\diamond_{i\in[t]}\sigma_{i})\star\overline{\sigma_t}\big)=
\varpi(\sigma_0^t)\bigg(\frac{\prod_{i\in[t+1]}\sigma_{i}}{2\varpi(\sigma_0^t)}
+\frac{\prod_{i\in[t+1]}\overline{\sigma_{i}}}{2\varpi(\sigma_0^t)}\bigg)
=\varpi(\sigma_0^{t+1}).\label{u10}
\end{gather}
\begin{lemma}
\label{lem700}
For $t\geq 1$ and $\varepsilon_0^t\in[0,1]^t$,
\begin{gather}
\Delta(\{\mathrm{B}(\varepsilon_i)\}_{i\in[t]})
\cong\mathrm{B}(\star_{i\in[t]}\varepsilon_i),\label{u08}\\
\nabla(\{\mathrm{B}(\varepsilon_i)\}_{i\in[t]})
\cong\sum_{\sigma_i\in\{\varepsilon_i,\overline{\varepsilon_i}\},i\in[t]}\varpi(\sigma_0^t)
\mathrm{B}(\diamond_{i\in[t]}\sigma_i),\label{u09}
\end{gather}
where $\varepsilon$ and $\overline{\varepsilon}$ should be seen as different elements even if $\varepsilon=1/2$.
\end{lemma}
\begin{proof}
Clearly, (\ref{u08}) follows from (\ref{at31}) and (\ref{u01}).
It is obvious that (\ref{u09}) is valid for $t=1$ or $\varepsilon_0^t\not\in(0,1)^t$.
Now we manage to prove (\ref{u09}) for $\varepsilon_0^t\in(0,1)^t$ by induction on $t$.
Indeed, according to (\ref{at01}), (\ref{at32}), (\ref{u02}) and (\ref{u10}) we have
\begin{align*}
&\nabla(\{\mathrm{B}(\varepsilon_i)\}_{i\in[t+1]})
=A_1(\nabla(\{\mathrm{B}(\varepsilon_i)\}_{i\in[t]}),\mathrm{B}(\varepsilon_t))\\
\cong&\sum_{\sigma_i\in\{\varepsilon_i,\overline{\varepsilon_i}\},i\in[t]}\varpi(\sigma_0^t)
A_1(\mathrm{B}(\diamond_{i\in[t]}\sigma_i),\mathrm{B}(\varepsilon_t))\\
\cong&\sum_{\sigma_i\in\{\varepsilon_i,\overline{\varepsilon_i}\},i\in[t]}\varpi(\sigma_0^t)
\sum_{\sigma_t\in\{\varepsilon_t,\overline{\varepsilon_t}\}}
((\diamond_{i\in[t]}\sigma_i)\star\overline{\sigma_t})\mathrm{B}(\diamond_{i\in[t+1]}\sigma_i)\\
\cong&\sum_{\sigma_i\in\{\varepsilon_i,\overline{\varepsilon_i}\},i\in[t+1]}\varpi(\sigma_0^{t+1})
\mathrm{B}(\diamond_{i\in[t+1]}\sigma_i).
\end{align*}

The proof is complete.
\end{proof}

The following theorem is a simple corollary of Theorem~\ref{lem200} and Lemma~\ref{lem700}.

\begin{theorem}\label{cor501'}
Let $\{W_l\}_{l\in [m]}$ be independent and symmetric BIDMCs.
Suppose $W_l\cong\sum_{i\in[n_l]}q_{i,l}\mathrm{B}(\varepsilon_{i,l})$ for $l\in[m]$, then we have
\begin{gather}
\Delta(\{W_l\}_{l\in[m]})\cong
\sum_{\{i_l\in[n_l]\}_{l\in[m]}}\bigg(\prod_{l\in[m]}q_{i_l,l}\bigg)
\mathrm{B}(\star_{l\in[m]}\varepsilon_{i_l,l}),
\label{mm01}\\
\nabla(\{W_l\}_{l\in[m]})\cong\sum_{\{i_l\in[n_l],\sigma_{l}\in
\{\varepsilon_{i_l},\overline{\varepsilon_{i_l}}\}\}_{l\in[m]}}
\bigg(\prod_{l\in [m]}q_{i_l,l}\bigg)\varpi(\sigma_0^m)\mathrm{B}(\diamond_{l\in[m]}\sigma_{l}),\label{mm02}
\end{gather}
where $\varepsilon$ and $\overline{\varepsilon}$ should be seen as different elements even if $\varepsilon=1/2$.
\end{theorem}

Notice that the synthetic BIDMCs $\Delta(\{W_l\}_{l\in[m]})$ and $\nabla(\{W_l\}_{l\in[m]})$ are equivalent to the parity-constrained-input parallel channel and the parallel broadcast channel discussed in \cite{Shamai05}, respectively, when the underlying channels $\{W_l\}_{l\in[m]}$ are symmetric BIDMCs.

\section{Synthetic Channels in Polar Codes}\label{sec04}

In this section, we mainly consider to express the synthetic channels, which are generated in the polar codes proposed by Arikan in
\cite{Arikan09}, as RSCs of only a few BSCs when the underlying channels are symmetric BIDMCs.

At first, we give a brief introduction for the polar codes from a new angle of view.
The synthetic channels in
the polar code of order $k\geq 2$ are Arikan transformations generated iteratively from underlying channels $\{W^{(0)}_{\alpha}\}_{\alpha\in \{0,1\}^k}$ which are independent copies of a given BIDMC $W$.
For any $0\leq m\leq k-2$, $a,b\in \{0,1\}$, $\beta\in \{0,1\}^{m}$ and $\alpha\in \{0,1\}^{k-m-2}$, let
\begin{align}
W^{(m+1)}_{\alpha\beta a b}=A_a(W^{(m)}_{\alpha b\beta0},W^{(m)}_{\alpha b\beta1}).
\end{align}
For any $\beta\in \{0,1\}^{k-1}$ and $a\in \{0,1\}$, let
\begin{align}
W^{(k)}_{\beta a}=A_a(W^{(k-1)}_{\beta0},W^{(k-1)}_{\beta1}).
\end{align}
Notice that, for any $0\leq m\leq k-1$ and $\beta\in\{0,1\}^m$, the channels in
$$\{W^{(m)}_{\alpha\beta a}:\alpha\in\{0,1\}^{k-1-m},a\in\{0,1\}\}$$ are independent and equivalent to each other.

If the tuples $\alpha=(d_{0},\ldots,d_{k-1})\in \{0,1\}^k$ are ordered according to the ascending order of $$b(\alpha)=\sum_{i\in [k]}d_i2^{k-1-i},$$
the synthetic channels generated in polar code of order $k=4$ are shown in the Figure~3, where $A_4$ and $A_8$ denote transformations with structures similar to that shown in the left sides respectively.

\setlength{\unitlength}{0.3cm}
\newcounter{dc0}
\newcounter{dc4}
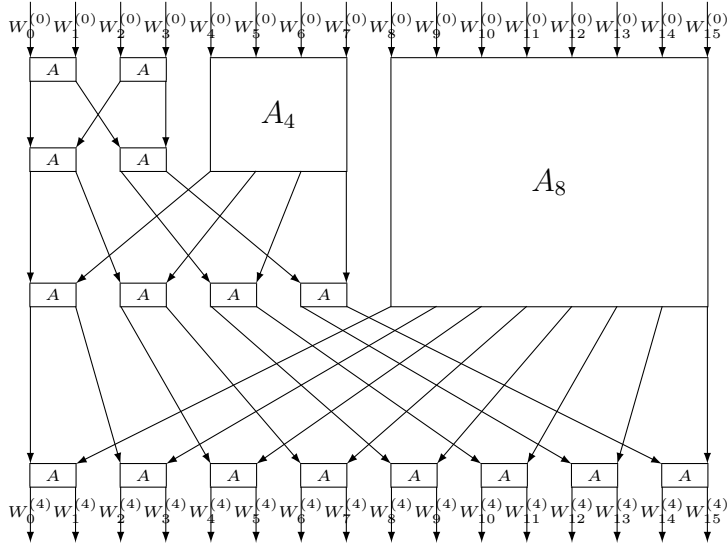
\begin{figure}[t]
\begin{center}
\ \ \ \ \ \ \begin{picture}(33,24)
\addtocounter{dc0}{-1}
\multiput(-1,21.9)(2,0){16}
{\addtocounter{dc0}{1}\makebox(0,0)[bl]{{\tiny $W^{(0)}_{\arabic{dc0}}$}}}
\addtocounter{dc4}{-1}
\multiput(-1,0.4)(2,0){16}
{\addtocounter{dc4}{1}\makebox(0,0)[bl]{{\tiny $W^{(4)}_{\arabic{dc4}}$}}}

\multiput(0,23.5)(2,0){16}{\vector(0,-1){2.5}}

\multiput(0,20)(4,0){2}{\framebox(2,1)[]{{\tiny $A$}}}
\multiput(0,16)(4,0){2}{\framebox(2,1)[]{{\tiny $A$}}}
\put(0,20){\vector(0,-1){3}}
\put(2,20){\vector(2,-3){2}}
\put(4,20){\vector(-2,-3){2}}
\put(6,20){\vector(0,-1){3}}

\multiput(0,10)(4,0){4}{\framebox(2,1)[]{{\tiny $A$}}}
\put(0,16){\vector(0,-1){5}}
\put(2,16){\vector(2,-5){2}}
\put(4,16){\vector(4,-5){4}}
\put(6,16){\vector(6,-5){6}}
\put(8,16){\vector(-6,-5){6}}
\put(10,16){\vector(-4,-5){4}}
\put(12,16){\vector(-2,-5){2}}
\put(14,16){\vector(0,-1){5}}
\put(8,16){\framebox(6,5){$A_4$}}

\multiput(0,2)(4,0){8}{\framebox(2,1)[]{{\tiny $A$}}}
\put(0,10){\vector(0,-1){7}}
\put(2,10){\vector(2,-7){2}}
\put(4,10){\vector(4,-7){4}}
\put(6,10){\vector(6,-7){6}}
\put(8,10){\vector(8,-7){8}}
\put(10,10){\vector(10,-7){10}}
\put(12,10){\vector(12,-7){12}}
\put(14,10){\vector(14,-7){14}}
\put(16,10){\vector(-14,-7){14}}
\put(18,10){\vector(-12,-7){12}}
\put(20,10){\vector(-10,-7){10}}
\put(22,10){\vector(-8,-7){8}}
\put(24,10){\vector(-6,-7){6}}
\put(26,10){\vector(-4,-7){4}}
\put(28,10){\vector(-2,-7){2}}
\put(30,10){\vector(0,-7){7}}
\put(16,10){\framebox(14,11){$A_8$}}
\multiput(0,2)(2,0){16}{\vector(0,-1){2.5}}

\end{picture}\end{center}
\caption{Synthetic channels in polar code of order $k=4$.}
\end{figure}

For $W_0\cong W_1\cong W$, we also write $A_0(W_0,W_1)$ and $A_1(W_0,W_1)$ as $A_0(W)$ and $A_1(W)$, respectively. For $n\geq 1$ and $\alpha\in\{0,1\}^n$, we define further $A_{\alpha 0}(W)=A_0(A_{\alpha}(W))$ and $A_{\alpha 1}(W)=A_1(A_{\alpha}(W))$ iteratively.
The following lemma is from \cite{Arikan09}.
\begin{lemma}
If the underlying channels $\{W^{(0)}_{\alpha}\}_{\alpha\in \{0,1\}^k}$ are copies of a given BIDMC $W$, then the set $\{W^{(k)}_{\alpha}:\alpha\in \{0,1\}^k\}$ of the synthetic channels is just the set $\{A_{\alpha }(W):\alpha\in \{0,1\}^k\}$
which shall exhibit polarization phenomena when $k$ is large: The capacity of the part with a proportion close to $I(W)$ is approximately equal to 1, and meanwhile the capacity of the part with a proportion close to $1-I(W)$ is approximately equal to 0.
\end{lemma}

For $d\in \{0,1\}$, let $\mathcal{S}(d)$ denote $\{1\}$ if $d=1$ and $\{0,1\}$ if $d=0$.
For $\delta=(d_0,\ldots,d_{n-1})\in \{0,1\}^n$, let
\begin{gather*}
\mathcal{S}(\delta)=\{(a_{0},\ldots,a_{n-1})\in \{0,1\}^n: a_i\in \mathcal{S}(d_{n-1-i}), i\in [n]\}.
\end{gather*}
For example, $\mathcal{S}(0101)$ is the set $\{1a1b:a,b\in\{0,1\}\}$.
If the inputs of the synthetic channels $\{W^{(k)}_{\alpha}\}_{\alpha\in\{0,1\}^{k}}$ are set to be $\{u_{\alpha}\}_{\alpha\in \{0,1\}^{k}}$, then those of the underlying channels $\{W^{(0)}_{\alpha}\}_{\alpha\in \{0,1\}^{k}}$ are given by
\begin{align}
\{x^{(0)}_{\alpha}\}_{\alpha\in \{0,1\}^{k}}=\{u_{\delta}\}_{\delta\in \{0,1\}^{k}}G_{k},
\label{ak05}
\end{align}
where $G_{k}$ is the 0\,-1 matrix of order $2^k$ whose entry at the $\alpha$-column and the $\delta$-row
is 1 if and only if $\delta\in\mathcal{S}(\alpha)$.
For example,
$$G_1=\left(
\begin{array}{cc}
1 & 0 \\
1 & 1 \\
\end{array}
\right),\
G_2=\left(
  \begin{array}{cccc}
    1 & 0 & 0 & 0 \\
    1 & 0 & 1 & 0 \\
    1 & 1 & 0 & 0 \\
    1 & 1 & 1 & 1 \\
  \end{array}
\right),\
G_3=\left(
  \begin{array}{cccccccc}
    1 & 0 & 0 & 0 & 0 & 0 & 0 & 0 \\
    1 & 0 & 0 & 0 & 1 & 0 & 0 & 0 \\
    1 & 0 & 1 & 0 & 0 & 0 & 0 & 0 \\
    1 & 0 & 1 & 0 & 1 & 0 & 1 & 0 \\
    1 & 1 & 0 & 0 & 0 & 0 & 0 & 0 \\
    1 & 1 & 0 & 0 & 1 & 1 & 0 & 0 \\
    1 & 1 & 1 & 1 & 0 & 0 & 0 & 0 \\
    1 & 1 & 1 & 1 & 1 & 1 & 1 & 1 \\
  \end{array}
\right).
$$

The {\it communication scheme} of
the polar code of order $k$ is as follows.
The information sequence $\{u_{\alpha}\}_{\alpha\in \{0,1\}^k}$ is coded into $\{x^{(0)}_{\alpha}\}_{\alpha\in \{0,1\}^k}$ according to (\ref{ak05}) and then transmitted over the BIDMCs $\{W^{(0)}_{\alpha}\}_{\alpha\in \{0,1\}^k}$ respectively, whereas the receiver decodes the synthetic channels $\{W^{(k)}_{\alpha}\}_{\alpha\in \{0,1\}^k}$ sequentially by the {\it Successive Cancellation} (SC) decoding, which generates estimates for $\{u_{\alpha}\}_{\alpha\in \{0,1\}^k}$ by MLD in successive order.
If a \emph{genie} is willing to show us the true inputs of the noisy synthetic channels, which are called {\it frozen channels}, then reliable communication can be realized by decoding only the remaining synthetic channels. Notice that the role of such a genie can be almost realized by a pseudo-random sequence of very large period, which is known at the two ends of communication.

To determine the frozen channels of polar codes, it is desired to compute the capacities of the synthetic channels $\{A_{\alpha }(W)\}_{\alpha\in \{0,1\}^k}$ for large $k$. We will consider to give the compactest RSC forms of these synthetic channels.

\subsection{Arikan Transformations $\Delta_{m}(W)$ and $\nabla_{m}(W)$}

If $\{W_i\}_{i\in[t]}$ are independent copies of a
symmetric BIDMC $W$, we also write $\Delta(\{W_i\}_{i\in[t]})$ and $\nabla(\{W_i\}_{i\in[t]})$
as $\Delta_t(W)$ and $\nabla_t(W)$, respectively. We will deal with these Arikan transformations in this subsection.

According to Lemma~\ref{lem301}, Theorem~\ref{lem200} and (\ref{u01}) to (\ref{u04}) one can deduce the following lemma easily.
\begin{lemma}
Suppose $W$ is a symmetric BIDMC. Let $\Delta_0(W)=\mathrm{B}(0)$ and $\nabla_0(W)=\mathrm{B}(1/2)$.
Then,

\noindent
1.
For nonnegative integers $i,j$, we have
\begin{gather}
A_0(\Delta_i(W),\Delta_j(W))=\Delta_{i+j}(W),\
A_1(\nabla_i(W),\nabla_j(W))=\nabla_{i+j}(W).\label{tt1}
\end{gather}

\noindent
2. For positive integer $k$, we have
\begin{gather}
A_{0^k}(W)=\Delta_{2^k}(W),\
A_{1^k}(W)=\nabla_{2^k}(W).\label{tt11}
\end{gather}
3. For nonnegative numbers $p,q$ with $p+q\leq 1$, let $W_{p,q}$ be
the symmetric BIDMC $p\mathrm{B}(0)+q\mathrm{B}(1/2)+\overline{p+q}W$. For positive integer $t$, we have
\begin{gather}
\Delta_{t}(W_{p,q})\cong \overline{\overline{q}^t}\mathrm{B}\Big(\frac{1}{2}\Big)+
\sum_{i=0}^{t}\binom{t}{i}p^{t-i}\overline{p+q}^i\Delta_{i}(W),\label{pp7}\\
\nabla_{t}(W_{p,q})\cong \overline{\overline{p}^t}\mathrm{B}(0)+
\sum_{i=0}^{t}\binom{t}{i}q^{t-i}\overline{p+q}^i\nabla_{i}(W).\label{pp8}
\end{gather}
\end{lemma}

The following corollary is a direct application of Theorem~\ref{cor501'}.
\begin{corollary}\label{lem501}
For $W\cong\sum_{j\in[n]}q_j\mathrm{B}(\varepsilon_j)$ with $\varepsilon_0^n\in (0,1)^n$ and $m\geq 1$,
\begin{gather}
\Delta_m(W)\cong\sum_{i_0^m\in[n]^m}\mathrm{B}(\star_{j\in[m]}\varepsilon_{i_j})\prod_{j\in[m]}q_{i_j},
\label{mm01}\\
\nabla_m(W)\cong\sum_{\{i_l\in[n],\sigma_{l}\in\{\varepsilon_{i_l},\overline{\varepsilon_{i_l}}\}\}_{
l\in[m]}}
\mathrm{B}(\diamond_{l\in[m]}\sigma_{l})\varpi(\sigma_0^m)\prod_{l\in [m]}q_{i_l},\label{mm02}
\end{gather}
where $\varepsilon$ and $\overline{\varepsilon}$ should be seen as different elements even if $\varepsilon=1/2$.
\end{corollary}

Clearly, according to Lemma \ref{lem201} one can express $\Delta_m(W)$ and $\nabla_m(W)$ as RSCs of less BSCs than those shown in Corollary~\ref{lem501} by merging the equivalent sub-channels.
For any symmetric BIDMC, its compactest expression of RSC of BSCs is the LRP-oriented form.
If the LRP-oriented form of $W$ is an RSC of $n$ BSCs, i.e., $W\cong\sum_{j\in[n]}q_j\mathrm{B}(\varepsilon_j)$ for some positive numbers $\{q_j\}_{j\in[n]}$ and $0\leq\varepsilon_0<\cdots<\varepsilon_{n-1}\leq 1/2$,
we denote $\phi(W)=n$.

\begin{lemma}\label{lemf00}
1. For any vector $a_0^n$ of nonnegative integers with $\sum_{i\in[n]}a_i=m$, let $K(a_0^n)$ denote the number of vectors in
$$\{d_0^m\in[n]^m:|\{j\in [m]: d_j=i\}|=a_i,\,i\in[n]\}.$$ Then,
\begin{gather} K(a_0^n)=\binom{m}{a_0^n}=\frac{m!}{\prod_{i\in[n]}(a_i!)}.\label{gg20}
\end{gather}
2. Let $H_{m,n}$ denote the number of methods for putting $m$ balls into $n$ distinct boxes. Then,
\begin{align}
H_{m,n}=\binom{m+n-1}{m}=\frac{(m+n-1)!}{m!(n-1)!}.
\end{align}
\end{lemma}
\begin{proof}
The proof for this lemma is simple, we omit it here.
\end{proof}
According to Corollary~\ref{lem501}, Lemmas~\ref{lem201} and \ref{lemf00}, one can deduce the following theorem easily.
\begin{theorem}\label{cor10}
For $m\geq 1$ and $W\cong\sum_{i\in[n]}q_i\mathrm{B}(\varepsilon_i)$ with $\varepsilon_0^n\in (0,1)^n$,

\noindent 1.
$\Delta_m(W)$ is equivalent to an RSC of $\phi(\Delta_m(W))\leq \binom{m+n-1}{m}$ BSCs as the following
\begin{align}
\Delta_m(W)\cong\sum_{\sum_{i\in[n]}a_i=m}
\mathrm{B}\big(\star_{i\in[n]}\varepsilon_{i}^{\star a_i}\big)\binom{m}{a_0^n}\prod_{i\in[n]}q_{i}^{a_i},\label{pp10}
\end{align}
where $a_0^n$ is a vector of nonnegative integers, $\varepsilon^{\star 0}=1$ and, for integer $a>0$,
$$\varepsilon^{\star a}=\underbrace{\varepsilon\star\cdots\star\varepsilon}_{a}
=\sum_{0\leq i\leq \lfloor (a-1)/2\rfloor}\binom{a}{2i+1}\varepsilon^{2i+1}\overline{\varepsilon}^{a-2i-1}.$$

\noindent 2. $\nabla_m(W)$ is equivalent to an RSC of
\begin{gather}
\phi(\nabla_m(W))\leq
1+\sum_{\omega=1}^{m}2^{\omega-1}\binom{n}{\omega}\sum_{0\leq 2b\leq m-\omega}\binom{m-2b-1}{m-2b-\omega}
\end{gather}
BSCs as the following
\begin{align}
\nabla_m(W)\cong&
\sum_{\Omega\subset[n],\sum_{i\in\Omega}a_i+2\sum_{j\in[n]}b_j=m}
\binom{m}{a_0^n+b_0^n,b_0^n}
\prod_{i\in[n]}q_{i}^{a_i+2b_i}(\varepsilon_i\overline{\varepsilon_i})^{b_i}\nonumber\\
&\ \ \ \ \ \ \ \ \ \ \
\sum_{\{\sigma_{l}\in\{\varepsilon_{l},\overline{\varepsilon}_{l}\}\}_{l\in\Omega}}
\frac{1}{2}\bigg(\prod_{i\in\Omega}\sigma_i^{a_i}+
\prod_{i\in\Omega}\overline{\sigma_i}^{a_i}\bigg)\mathrm{B}(\diamond_{i\in\Omega}\sigma^{\diamond a_i}_{i}),
\label{pp4}
\end{align}
where $\varepsilon$ and $\overline{\varepsilon}$ should be always seen as different elements, $b_j\geq 0$ for $j\in[n]$, $a_i>0$ for $i\in\Omega$, $a_l=0$ for $l\in[n]\setminus\Omega$,
$$\sigma^{\diamond a}=\underbrace{\sigma\diamond\cdots\diamond\sigma}_{a}=
\frac{\sigma^a}{\sigma^a+\overline{\sigma}^a},$$
and the inner summation in (\ref{pp4}) denotes $\mathrm{B}(1/2)$ if $\Omega=\emptyset$.
\end{theorem}
\begin{proof}
(\ref{pp10}) and (\ref{pp4}) follow
simply from Corollary~\ref{lem501}, Lemmas \ref{lem201} and \ref{lemf00}.
According to (\ref{pp10}) we see $\phi(\Delta_m(W))\leq H_{m,n}=\binom{m+n-1}{m}$.
According to (\ref{pp4}) and $\mathrm{B}(\diamond_{i\in\Omega}\sigma^{\diamond a_i}_{i})\cong\mathrm{B}(\diamond_{i\in\Omega}\overline{\sigma}^{\diamond a_i}_{i})$ we see
\begin{align*}
\phi(\nabla_m(W))\leq& 1+\sum_{\omega=1}^{m}2^{\omega-1}\binom{n}{\omega}\sum_{0\leq 2b\leq m-\omega}H_{m-\omega-2b,\omega}\\
=&1+\sum_{\omega=1}^{m}2^{\omega-1}\binom{n}{\omega}\sum_{0\leq 2b\leq m-\omega}\binom{m-2b-1}{m-2b-\omega},
\end{align*}
where $\omega$ and $b$ correspond the cardinality of the set $\Omega$ and the summation $\sum_{i\in[n]}b_i$, respectively.
\end{proof}

If some of the cross-over probabilities are equal to $1/2$, then Theorem \ref{cor10}
can be refined further as the following theorem.
\begin{theorem}
\label{lem80}
For $m\geq 1$ and $W\cong\sum_{j\in[n+1]}q_j\mathrm{B}(\varepsilon_j)$ with $0<\varepsilon_0<\cdots<\varepsilon_{n-1}<\varepsilon_n=1/2$,

\noindent 1.
$\Delta_{m}(W)$ is equivalent to an RSC of $\phi(\Delta_{m}(W))\leq H_{m,n-1}+1=\binom{m +n-1}{m}+1$ BSCs as the following
\begin{align}
\Delta_{m}(W)\cong&
(1-\overline{q_n}^{m})\mathrm{B}\big(\frac{1}{2}\big)+\sum_{\sum_{i\in[n]}a_i=m}
\mathrm{B}\big(\star_{i\in[n]}\varepsilon_{i}^{\star a_i}\big)
\binom{m}{a_0^n}\prod_{i\in[n]}q_{i}^{a_i},
\label{pp11}
\end{align}
where $a_0^n$ is vector of nonnegative integers.

\noindent 2. $\nabla_{m}(W)$ is equivalent to an RSC of
\begin{gather}
\phi(\nabla_{m}(W))\leq 1+\sum_{\omega=1}^{m}2^{\omega-1}\binom{n}{\omega}\binom{m}{\omega}
\end{gather}
BSCs as the following
\begin{align}
\nabla_{m}(W)\cong&
\sum_{s+\sum_{j\in[n]}(a_j+2b_j)=m}
\binom{m}{s,a_0^n+b_0^n,b_0^n}
q_n^s\prod_{i\in[n]}q_{i}^{a_i+2b_i}(\varepsilon_i\overline{\varepsilon_i})^{b_i}\nonumber\\
&\ \ \ \ \ \
\sum_{\{\sigma_l\in\{\varepsilon_l,\overline{\varepsilon}_l\}\}_{a_l>0,l\in[n]}}
\frac{1}{2}\bigg(\prod_{i\in[n]}\sigma_i^{a_i}+\prod_{i\in[n]}\overline{\sigma_i}^{a_i}\bigg)
\mathrm{B}(\diamond_{i\in[n]}\sigma^{\diamond a_i}_{i}),\label{pp5}
\end{align}
where integer $s\geq 0$, $a_0^n,b_0^n$ are vectors of nonnegative integers, $\sigma^{\diamond 0}=1/2$ and the inner summation in (\ref{pp5}) denotes $\mathrm{B}(1/2)$ if $\sum_{i\in[n]}a_i=0$.
\end{theorem}
\begin{proof}
Clearly, (\ref{pp11}) follows from (\ref{pp7}) and (\ref{pp10}).
For any vectors $a_0^n, b_0^n$ of nonnegative integers with $\sum_{i\in[n]}(a_i+2b_i)\leq m$,
let $s=m-\sum_{i\in[n]}(a_i+2b_i)$. If $s=2t+1$ is odd, then
we have
\begin{align*}
&2\sum_{b_n=0}^t\binom{m}{a_0^n+b_0^n,b_0^n}=
2^{s}\binom{m}{s,a_0^n+b_0^n,b_0^n},
\end{align*}
where $a_n=s-2b_n$ is always positive. If $s=2t$ is even, then we have
\begin{align*}
&\binom{m}{a_0^n+b_0^n,b_0^n,t,t}+
2\sum_{b_n=0}^{t-1}\binom{m}{a_0^{n+1}+b_0^{n+1},b_0^{n+1}}
=2^{s}\binom{m}{s,a_0^n+b_0^n,b_0^n},
\end{align*}
where $a_n=s-2b_n$ is always positive too. Therefore, we see that (\ref{pp5}) follows from (\ref{pp4}), $\varepsilon_n=\frac{1}{2}$ and $\sigma\diamond \frac{1}{2}=\sigma$ for any $\sigma\in(0,1)$.
Furthermore, from (\ref{pp5}) and $\overline{\diamond_{i\in[n]}\sigma^{\diamond a_i}_{i}}=\diamond_{i\in[n]}\overline{\sigma}^{\diamond a_i}_{i}$ we see
\begin{align*}
\phi(\nabla_m(W))\leq& 1+\sum_{\omega=1}^{m}2^{\omega-1}\binom{n}{\omega}\sum_{0\leq k\leq m-\omega}H_{m-\omega-k,\omega}\\
=&1+\sum_{\omega=1}^{m}2^{\omega-1}\binom{n}{\omega}\sum_{0\leq k\leq m-\omega}\binom{m-k-1}{m-k-\omega}\\
=&1+\sum_{\omega=1}^{m}2^{\omega-1}\binom{n}{\omega}\binom{m}{\omega},
\end{align*}
where $\omega$ corresponds the number of positive integers $a_i$, $k$ corresponds the sum of $s$ and $2\sum_{i\in[n]}b_i$.
\end{proof}

We note that, for any synthetic channel generated in polar codes over symmetric BIDMCs,
it is not difficult to give an exact expression of RSC of BSCs for it by using
Theorems~\ref{cor10} and \ref{lem80}.

\subsection{Arikan Transformations $A_{\alpha}(\mathrm{E}(q))$ and $A_{\alpha}(\mathrm{B}(\varepsilon))$}

It has been pointed out in \cite{Arikan09} that all the Arikan transformations are BECs when
the underlying channel is a BEC.
Indeed, if the underlying channel $W$ is $\mathrm{E}(q)\cong\overline{q}\mathrm{B}(0)+q\mathrm{B}(1/2)$, the BEC with erasure probability $q$, then for any integer $t$ from (\ref{pp7}) and (\ref{pp8}) we have
\begin{gather}
\Delta_{t}(\mathrm{E}(q))\cong
\overline{q}^t\mathrm{B}(0)+\overline{\overline{q}^t}\mathrm{B}(1/2)
\cong\mathrm{E}(\overline{\overline{q}^t}), \\
\nabla_{t}(\mathrm{E}(q))\cong
\overline{q^t}\mathrm{B}(0)+q^t\mathrm{B}(1/2)\cong\mathrm{E}(q^t).
\end{gather}
Furthermore, for nonnegative integer $s$, let $f_s$ denote the map over $[0,1]$ defined by $f_s(p)=\overline{p}^{2^s}$.
Therefore, for any sequence $\alpha=0^{t_1}1^{t_2}\cdots0^{t_{2r-1}}1^{t_{2r}}$ in $\{0,1\}^{*}$ with $t_1\geq 0,t_{2r}\geq 0$ and $t_i>0, 2\leq i\leq 2r-1$, we have
\begin{gather*}
A_{\alpha}(\mathrm{E}(q))\cong \overline{F_{\alpha}(q)}\mathrm{B}(0)
+F_{\alpha}(q)\mathrm{B}(1/2)
\cong \mathrm{E}(F_{\alpha}(q)),\\
I(A_{\alpha}(\mathrm{E}(q)))=\overline{F_{\alpha}(q)},
\end{gather*}
where $F_{\alpha}$ is the compound map $f_{t_{2r}}\circ f_{t_{2r-1}}\circ\cdots \circ f_{t_1}$ over $[0,1]$.
For example, for $\alpha=0110=0^11^20^11^0$, we have
\begin{gather*}
F_{0110}(q)=1-(1-(1-(1-q)^2)^4)^2,\\
I(A_{0110}(\mathrm{E}(q)))=\overline{F_{0110}(q)}=(1-(1-(1-q)^2)^4)^2.
\end{gather*}

If the underlying channel $W$ is a general BIDMC but BEC,
for any Arikan transformation of $W$ one can deduce an exact expression of RSC of BSCs by using
the results shown in the last subsection.
However, as pointed in the proof of Theorem~\ref{lem80}, the Arikan transformation $\nabla_{m}(W)$ given in (\ref{pp5}) is not in LRP-oriented form, some of the sub-channels may be merged further.

In the following theorem we show the LRP-oriented forms for some of the synthetic channels generated in polar codes when the underlying channel is a BSC.

\begin{theorem}
For $l\geq 0$ and $\varepsilon\in(0,1/2)$, we have
\begin{gather}
A_{0^l}(\mathrm{B}(\varepsilon))\cong \mathrm{B}(\varepsilon_l),\label{mm307'}
\end{gather}
where $\varepsilon_l=\varepsilon^{\star 2^l}$. For $k\geq 0$, we have
\begin{align}
A_{0^l1^{k+1}}(\mathrm{B}(\varepsilon))\cong
\binom{2^{k+1}}{2^{k}}(\varepsilon_l\overline{\varepsilon_l})^{2^{k}}
\mathrm{B}\Big(\frac{1}{2}\Big)+
\sum_{i=1}^{2^{k}}\binom{2^{k+1}}{2^{k}-i}
\frac{\varepsilon_l^{2i}+\overline{\varepsilon_l}^{2i}}
{(\varepsilon_l\overline{\varepsilon_l})^{i-2^k}}
\mathrm{B}(\varepsilon_l^{\diamond (2i)}).\label{pp0}
\end{align}
For $i\geq 0$, let $\varepsilon_{l,i}=(\varepsilon_l^{\diamond 2})^{\star 2^i}$ and $q_{l,i}
=(\varepsilon_l^{2}+\overline{\varepsilon_l}^{2})^{2^i}$. Then, for $k\geq 0$, we have
\begin{align}
A_{0^l10^i1^k}(\mathrm{B}(\varepsilon))
\cong&
\sum_{0\leq b\leq 2^{k-1}}
\binom{2^k}{2^k-2b,b,b}
\overline{q_{l,i}}^{2^k-2b}(q_{l,i}^2\varepsilon_{l,i}\overline{\varepsilon_{l,i}})^{b}
\mathrm{B}\Big(\frac{1}{2}\Big)\nonumber\\
+\sum_{a=1}^{2^{k}}\mathrm{B}(\varepsilon_{l,i}^{\diamond a})&
(\varepsilon_{l,i}^{a}+\overline{\varepsilon_{l,i}}^{a})q_{l,i}^a
\sum_{s+2b=2^k-a}
\binom{2^k}{s,a+b,b}
\overline{q_{l,i}}^s(q_{l,i}^2\varepsilon_{l,i}\overline{\varepsilon_{l,i}})^{b}.\label{pp1'''}
\end{align}
Let
$p_{l,i}=q_{l,i}^2(\varepsilon_{l,i}^2+\overline{\varepsilon_{l,i}}^2)+2q_{l,i}\overline{q_{l,i}}$ and
$r_{l,i}=2q_{l,i}\overline{q_{l,i}}/p_{l,i}$. Then, for $t\geq 0$, we have
\begin{align}
&A_{0^l10^i10^t}(\mathrm{B}(\varepsilon))\cong
a_{l,i,t}\mathrm{B}\Big(\frac{1}{2}\Big)+
\sum_{s=0}^{2^t}b_{l,i,t,s}
\mathrm{B}\big(\beta_{l,i,t,s}\big),\label{pp2''}
\end{align}
where $a_{l,i,t}=1-p_{l,i}^{2^t}$,
$b_{l,i,t,s}=\binom{2^t}{s}\overline{r_{l,i}}^{2^t-s}r_{l,i}^{s}p_{l,i}^{2^t}$ and
$\beta_{l,i,t,s}=(\varepsilon_{l,i}^{\diamond 2})^{\star(2^t-s)}\star(\varepsilon_{l,i})^{\star s}$ for $0\leq s\leq 2^t$. Furthermore, we have $0<\beta_{l,i,t,0}<\beta_{l,i,t,1}<\cdots<\beta_{l,i,t,2^t}<1/2$ and
\begin{align}
&A_{0^l10^i10^t1}(\mathrm{B}(\varepsilon))
\cong
\Big(a_{l,i,t}^2+2\sum_{s=0}^{2^t}b_{l,i,t,s}^2\beta_{l,i,t,s}\overline{\beta_{l,i,t,s}}\Big)
\mathrm{B}\Big(\frac{1}{2}\Big)+\nonumber\\
&\ \ \ \ 2a_{l,i,t}\sum_{s=0}^{2^t}b_{l,i,t,s}\mathrm{B}(\beta_{l,i,t,s})
+\sum_{s=0}^{2^t}b_{l,i,t,s}^2(\beta_{l,i,t,s}^2+\overline{\beta_{l,i,t,s}}^2)
\mathrm{B}(\beta_{l,i,t,s}^{\diamond 2})+
\nonumber\\
&\ \ \ \ 2\sum_{0\leq s<r\leq 2^t}b_{l,i,t,s}b_{l,i,t,r}\sum_{\sigma\in\{\beta_{l,i,t,r},\overline{\beta_{l,i,t,r}}\}}
(\beta_{l,i,t,s}\sigma+\overline{\beta_{l,i,t,s}}\,\overline{\sigma})
\mathrm{B}(\beta_{l,i,t,s}\diamond\sigma).
\label{h001}
\end{align}
\end{theorem}
\begin{proof}
From (\ref{tt11}) and (\ref{pp10}), we see (\ref{mm307'}).

From (\ref{tt11}), (\ref{pp4}) and (\ref{mm307'}), we see
\begin{align*}
&A_{0^l1^{k+1}}(\mathrm{B}(\varepsilon)) \cong A_{1^{k+1}}(A_{0^l}(\mathrm{B}(\varepsilon))) \cong A_{1^{k+1}}(\mathrm{B}(\varepsilon_l))\cong\nabla_{2^{k+1}}(\mathrm{B}(\varepsilon_l))\\
\cong&
\sum_{a+2b=2^{k+1}}\binom{2^{k+1}}{a+b,b}(\varepsilon_l\overline{\varepsilon_l})^b
\sum_{\sigma\in\{\varepsilon_l,\overline{\varepsilon_l}\}}\frac{\sigma^a+\overline{\sigma}^a}{2}
\mathrm{B}(\sigma^{\diamond a}),
\end{align*}
where the inner summation denotes $\mathrm{B}(1/2)$ when $a=0$.
Therefore, by taking $a=2i$ and $b=2^k-i$ we see that (\ref{pp0}) follows from $\mathrm{B}(\sigma^{\diamond a})\cong\mathrm{B}(\overline{\sigma}^{\diamond a})$.

From (\ref{pp7}), (\ref{mm307'}) and (\ref{pp0}) we have
\begin{align*}
&A_{0^l10^i1^k}(\mathrm{B}(\varepsilon))
\cong A_{0^i1^k}(A_{0^l1}(\mathrm{B}(\varepsilon)))\cong
A_{0^i1^k}(2\varepsilon_l\overline{\varepsilon_l}\mathrm{B}(1/2)
+(\varepsilon_l^2+\overline{\varepsilon_l}^2)\mathrm{B}(\varepsilon_l^{\diamond 2}))\nonumber\\
\cong &
A_{1^k}(\overline{q_{l,i}}\mathrm{B}(1/2)+q_{l,i}\mathrm{B}(\varepsilon_{l,i}))
\cong \nabla_{2^k}(\overline{q_{l,i}}\mathrm{B}(1/2)+q_{l,i}\mathrm{B}(\varepsilon_{l,i}))
\nonumber\\
\cong&
\sum_{0\leq b\leq 2^{k-1}}
\binom{2^k}{2^k-2b,b,b}
\overline{q_{l,i}}^{2^k-2b}q_{l,i}^{2b}(\varepsilon_{l,i}\overline{\varepsilon_{l,i}})^{b}
\mathrm{B}\Big(\frac{1}{2}\Big)\nonumber\\
&+\sum_{a=1}^{2^{k}}\sum_{\sigma\in\{\varepsilon_{l,i},\overline{\varepsilon_{l,i}}\}}
\frac{1}{2}\bigg(\sigma^{a}+\overline{\sigma}^{a}\bigg)
\mathrm{B}(\sigma^{\diamond a})\sum_{s+2b=2^k-a}
\binom{2^k}{s,a+b,b}
\overline{q_{l,i}}^sq_{l,i}^{a+2b}(\varepsilon_{l,i}\overline{\varepsilon_{l,i}})^{b}
\nonumber\\
\cong&
\sum_{0\leq b\leq 2^{k-1}}
\binom{2^k}{2^k-2b,b,b}
\overline{q_{l,i}}^{2^k-2b}(q_{l,i}^2\varepsilon_{l,i}\overline{\varepsilon_{l,i}})^{b}
\mathrm{B}\Big(\frac{1}{2}\Big)\nonumber\\
&+\sum_{a=1}^{2^{k}}\mathrm{B}(\varepsilon_{l,i}^{\diamond a})
(\varepsilon_{l,i}^{a}+\overline{\varepsilon_{l,i}}^{a})q_{l,i}^a
\sum_{s+2b=2^k-a}
\binom{2^k}{s,a+b,b}
\overline{q_{l,i}}^s(q_{l,i}^2\varepsilon_{l,i}\overline{\varepsilon_{l,i}})^{b},
\end{align*}
i.e. (\ref{pp1'''}) is valid for $k\geq 0$.

From (\ref{pp7}), (\ref{pp10}) and (\ref{pp1'''}) we have
\begin{align*}
&A_{0^l10^i10^t}(\mathrm{B}(\varepsilon))=A_{0^t}(A_{0^l10^i1}(\mathrm{B}(\varepsilon)))\nonumber\\
\cong& A_{0^t}\Big(q_{l,i}^2(\varepsilon_{l,i}^2+\overline{\varepsilon_{l,i}}^2)
\mathrm{B}(\varepsilon_{l,i}^{\diamond 2})
+2q_{l,i}\overline{q_{l,i}}\mathrm{B}(\varepsilon_{l,i})+
(2\varepsilon_{l,i}\overline{\varepsilon_{l,i}}q_{l,i}^2+\overline{q_{l,i}}^2)
\mathrm{B}(1/2)\Big)\nonumber\\
\cong&\Delta_{2^t}\Big(p_{l,i}\big(\overline{r_{l,i}}\mathrm{B}(\varepsilon_{l,i}^{\diamond 2})
+r_{l,i}\mathrm{B}(\varepsilon_{l,i})\big)+\overline{p_{l,i}}\mathrm{B}(1/2)\Big)\nonumber\\
\cong&(1-p_{l,i}^{2^t})\mathrm{B}(1/2)+p_{l,i}^{2^t}\Delta_{2^t}
\big(\overline{r_{l,i}}\mathrm{B}(\varepsilon_{l,i}^{\diamond 2})
+r_{l,i}\mathrm{B}(\varepsilon_{l,i})\big)\nonumber\\
\cong &
\big(1-p_{l,i}^{2^t}\big)\mathrm{B}\Big(\frac{1}{2}\Big)+p_{l,i}^{2^t}
\sum_{s=0}^{2^t}\binom{2^t}{s}\overline{r_{l,i}}^{2^t-s}r_{l,i}^{s}
\mathrm{B}\big((\varepsilon_{l,i}^{\diamond 2})^{\star (2^t-s)}\star\varepsilon_{l,i}^{\star s}\big),
\end{align*}
i.e. (\ref{pp2''}) is valid for any $t\geq 0$.

Since for any $\alpha,\sigma\in(0,1/2)$ we have $$1/2>\alpha\star\sigma>\max\{\alpha,\sigma\}\geq\min\{\alpha,\sigma\}>\alpha\diamond\sigma>0,$$ we see easily $0<\beta_{l,i,t,0}<\beta_{l,i,t,1}<\cdots<\beta_{l,i,t,2^t}<1/2$.
Therefore, (\ref{h001}) follows from (\ref{at01}), (\ref{at32}) and (\ref{pp2''}).
\end{proof}

Notice that, by taking $i=0$ in (\ref{pp1'''}), from (\ref{pp0}) one can deduce easily that,
for $k\geq 0$ and $0\leq a\leq 2^k$,
\begin{gather}
\sum_{s+2b=2^k-a}\binom{2^k}{s,a+b,b}2^{s}=\binom{2^{k+1}}{2^k-a},
\end{gather}
where $s, b$ assume nonnegative integers.

\subsection{Number of BSCs in Arikan Transformations}
Let $0\leq\varepsilon_0<\cdots<\varepsilon_{n-1}<\varepsilon_n=\frac{1}{2}$ and
$W\cong\sum_{j\in[n+1]}q_j\mathrm{B}(\varepsilon_j)$.
Clearly, the sizes of the output sets of $W$, $A_0(W)$, $A_1(W)$, $A_{00}(W)$, $A_{01}(W)$, $A_{10}(W)$, $A_{11}(W)$ are at least $2n+1$, $(2n+1)^2$, $2(2n+1)^2$, $(2n+1)^4$, $2(2n+1)^4$, $4(2n+1)^4$, $8(2n+1)^4$, respectively.
However, on the numbers of RSCs in their LRP-oriented forms, according to Theorem~\ref{lem80}, we see
\begin{align}\label{pp60}
\phi(A_0(W))&\leq \frac{n^2+n}{2}+1,\\
\phi(A_1(W))&\leq n^2+n+1,\label{pp61}\\
\phi(A_{00}(W))&\leq \binom{n+3}{4}+1=\frac{(n^2+n)(n^2+5n+6)}{24}+1,\label{pp62}\\
\phi(A_{11}(W))&\leq 4\binom{n}{1}+12\binom{n}{2}+16\binom{n}{3}+8\binom{n}{4}+1\nonumber\\
&=\frac{(n^2+n)(n^2+n+4)}{3}+1,\label{pp63}
\end{align}
and by using (\ref{pp60}) and (\ref{pp61}) we have
\begin{align}\label{pp66}
\phi(A_{01}(W))&\leq \frac{(n^2+n)(n^2+n+2)}{4}+1,\\
\phi(A_{10}(W))&\leq \frac{(n^2+n)(n^2+n+1)}{2}+1.
\end{align}

Moreover, for $\alpha\in \{0,1\}^k$, let $\varphi(\alpha)$ denote the number defined by $\varphi(1^l)=\varphi(0^l)=(2^l)!2^{2^l}$ for $l\geq 0$ and, for $\beta\in \{0,1\}^s$,
\begin{gather*}
\varphi(\beta 01^l)=(2^l)!(\varphi(\beta 0))^{2^l},\
\varphi(\beta 10^l)=(2^l)!(\varphi(\beta 1))^{2^l}.
\end{gather*}
Then, we have
\begin{theorem}
Let $0\leq\varepsilon_0<\cdots<\varepsilon_{n-1}<\varepsilon_n=\frac{1}{2}$ and
$W\cong\sum_{j\in[n+1]}q_j\mathrm{B}(\varepsilon_j)$.
Then, for $\alpha\in \{0,1\}^k$, we have
\begin{align}
\label{f001}
\phi(A_{\alpha}(W))\leq h_{\alpha}(n)=\frac{2^{b(\alpha)}(2n)^{2^k}}{\varphi(\alpha)}+g_{\alpha}(n),
\end{align}
where $g_{\alpha}(n)$ is a polynomial of order at most $2^k-1$ and
$b(\alpha)=a_{k-1}2^{k-1}+a_{k-2}2^{k-2}+\cdots+a_{0}$ if $\alpha=a_{k-1}a_{k-2}\cdots a_0$.
In particular, for any $\alpha\in \{0,1\}^k$ the average number of elements with the same likelihood ratios in the output set of $A_{\alpha}(W)$
is at least $\varphi(\alpha)/2$ when $n$ is sufficiently large.
\end{theorem}
\begin{proof}
If $\alpha=0^k$, then we have $b(\alpha)=0$, $\varphi(\alpha)=(2^k)!2^{2^k}$ and, according to (\ref{tt11}) and Theorem~\ref{lem80},
there is a polynomial $g_{0^k}(n)$ of order at most $2^k-1$ such that
\begin{align*}
\phi(A_{\alpha}(W))\leq 1+\binom{n+2^k-1}{2^k}=\frac{n^{2^k}}{(2^k)!}+g_{0^k}(n)
=\frac{2^{b(\alpha)}(2n)^{2^k}}{\varphi(\alpha)}+g_{0^k}(n),
\end{align*}
i.e., (\ref{f001}) is valid for $\alpha=0^k$.

If $\alpha=1^k$, then we have $b(\alpha)=2^k-1$, $\varphi(\alpha)=(2^k)!2^{2^k}$ and, according to (\ref{tt11}) and Theorem~\ref{lem80},
there is a polynomial $g_{1^k}(n)$ of order at most $2^k-1$ such that
\begin{align*}
\phi(A_{\alpha}(W))&\leq
1+\sum_{\omega=1}^{2^k}2^{\omega-1}\binom{n}{\omega}\binom{2^k}{\omega}
=\frac{2^{2^k-1}n^{2^k}}{(2^k)!}+g_{1^k}(n)
=\frac{2^{b(\alpha)}(2n)^{2^k}}{\varphi(\alpha)}+g_{1^k}(n),
\end{align*}
i.e., (\ref{f001}) is valid for $\alpha=1^k$.

Now we consider to prove (\ref{f001}) by induction on $k$.
Assume that (\ref{f001}) is valid for any sequence $\alpha$ of length smaller than $k$.
Suppose $\alpha\in \{0,1\}^k\setminus\{0^k,1^k\}$ and $l$ is the largest integer
such that $\alpha=\sigma\delta$, $\delta\in\{0^l,1^l\}$ and $\sigma\in\{0,1\}^{k-l}$. Then, we have $1\leq l<k$,
$\varphi(\delta)=(2^l)!2^{2^l}$, $\varphi(\alpha)=(2^l)!(\varphi(\sigma))^{2^l}$, $b(\alpha)=b(\sigma)2^l+b(\delta)$ and
\begin{gather*}
\phi(A_{\sigma}(W))\leq h_{\sigma}(n)=\frac{2^{b(\sigma)}(2n)^{2^{k-l}}}{\varphi(\sigma)}+g_{\sigma}(n).
\end{gather*}
Hence, we see there is a polynomial $g_{\alpha}(n)$ of order at most $2^k-1$ such that
\begin{align*}
\phi(A_{\alpha}(W))&=\phi(A_{\delta}(A_{\sigma}(W)))\\
&\leq \frac{2^{b(\delta)}(2h_{\sigma}(n))^{2^l}}{\varphi(\delta)}+g_{\delta}(h_{\sigma}(n))\\
&=\frac{2^{b(\delta)}2^{2^l}2^{b(\sigma)2^l}(2n)^{2^{k-l}2^l}}{(2^l)!2^{2^l}(\varphi(\sigma))^{2^l}}+g_{\alpha}(n)\\
&=\frac{2^{b(\alpha)}(2n)^{2^k}}{\varphi(\alpha)}+g_{\alpha}(n),
\end{align*}
where $g_{\delta}(h_{\sigma}(n))$ is a polynomial of order at most
$(2^l-1)2^{k-l}\leq 2^k-1$. Thus, (\ref{f001}) is valid for any sequence $\alpha$ of length $k$.

Since the size of the output set of $A_{\alpha}(W)$ is $2^{b(\alpha)}(2n+1)^{2^k}$, from (\ref{f001}) we see that the average size of the non-empty sets $L_{A_{\alpha}(W)}(\varepsilon)$
is at least $\varphi(\alpha)/2$ when $n$ is sufficiently large.
\end{proof}

Notice that one can deduce easily that, for each $\alpha\in \{0,1\}^k$, the number $\varphi(\alpha)2^{1-2^{k+1}}$ is an odd integer and
\begin{align*}
2^{2^{k+1}-1}\leq\varphi(\alpha)\leq (2^k)!2^{2^k}.
\end{align*}

\section{Conclusions}
\label{sec05}
The focus of this paper is to investigate the principal properties of the synthetic channels that are iteratively formed through Arikan transformations during the construction of polar codes. Given that evaluating the reliability of these synthetic channels is of great significance.

Utilizing the likelihood ratio profile (LRP), we established the equivalence and symmetry of binary input discrete memoryless channels (BIDMCs). By representing symmetric BIDMCs as random switching channels (RSCs) of binary symmetric channels (BSCs), we managed to convert channel transformations into algebraic operations.

By taking advantage of these algebraic operations, we obtained concise and compact expressions for the Arikan transformations of symmetric BIDMCs. In the case where the underlying channel is a BSC, we proposed the LRP-oriented forms of RSCs for several synthetic channels generated in the process of polar code construction.

In the end, we derived a lower bound for the average number of elements that possess the same likelihood ratio within the output alphabet of any synthetic channel in a polar code.


\begin{thebibliography}{10}
	
\bibitem{Shamai05} I. Sutskover, S. Shamai (Shitz) and J. Ziv, ``Extremes of information combining," {\it IEEE Trans. Inf. Theory,} vol. 51, no. 4, pp. 1313-1325, Apr. 2005.

\bibitem{Urbanke08} T. Richardson and R. Urbanke, Modern Coding Theory. Cambridge, U.K.: Cambridge Univ. Press, 2008.

\bibitem{Arikan09} E. Arikan, ``Channel polarization: A method for constructing capacity-achieving codes for symmetric binary-input memoryless channels," {\it IEEE Trans. Inf. Theory,} vol. 55, no. 7, pp. 3051-3073, Jul. 2009.

\bibitem{Tanaka09} R. Mori and T. Tanaka, ``Performance of polar codes with the construction using density evolution," {\it IEEE Commun. Lett.,} vol. 13, no. 7, pp. 519-521, Jul. 2009.
\bibitem{Korada09} S. B. Korada, ``Polar codes for channel and source coding," Ph.D. dissertation, Ecole Polytechnique Federale de Lausanne, Lausanne, Switzerland, 2009.
\bibitem{Trifonov12} P. Trifonov, ``Efficient design and decoding of polar codes," {\it IEEE Trans. Commun.,} vol. 60, no. 11, pp. 3221-3227, Nov. 2012.
\bibitem{Vardy13} I. Tal and A. Vardy, ``How to construct polar codes," {\it IEEE Trans. Inf. Theory,} vol. 59, no. 10, pp. 6562-6582, Oct. 2013.

\bibitem{Schurch16} C. Schurch, ``A partial order for the synthesized channels of a polar code," {\it in Proc. IEEE Int. Symp. Inf. Theory (ISIT),} Barcelona, Spain, pp. 220-224, Jul. 2016.
\bibitem{3GPP16} 3GPP, R1-167209, Polar code design and rate matching, Huawei, HiSilicon.

\bibitem{Siegel17} M. Qin, J. Guo, A. Bhatia, A. G. I. Faabregas and P. Siegel, ``Polar code constructions based on LLR evolution," {\it IEEE Commun. Lett.,} vol. 21, no. 6, pp. 1221-1224, Jun. 2017.
\bibitem{He17} G. He et al., ``$\beta$-expansion: A theoretical framework for fast and
recursive construction of polar codes," {\it in Proc. IEEE Global Commun.
Conf.,} pp. 1-6, Dec. 2017.

\bibitem{Ye18} T. C. Gulcu, M. Ye and A. Barg, ``Construction of polar codes
for arbitrary discrete memoryless channels," {\it IEEE Trans. Inf. Theory,}
vol. 64, no. 1, pp. 309-321, Jan. 2018.

\bibitem{Hirche18} C. Hirche and D. Reeb, ``Bounds on information combining with quantum side information," {\it IEEE Trans. Inf. Theory,} vol. 64, no. 7, pp. 4739-4757, Jul. 2018.


\bibitem{Urbanke19} M. Mondelli, S. H. Hassani and R. L. Urbanke,
``Construction of polar codes with sublinear complexity,"
{\it IEEE Trans. Inf. Theory,} vol. 65, no. 5, pp. 2782-2791, May 2019.
\bibitem{Siegel19} W. Wu and P. H. Siegel, ``Generalized partial orders for polar code bit-channels," {\it IEEE Trans. Inf. Theory,} vol. 65, no. 11, pp. 7114-7130, Nov. 2019.
\bibitem{Lin20} C. Lin, Y. Huang, S. Shieh and P. Chen, ``Transformation of binary linear block codes to polar codes with dynamic frozen," {\it IEEE Open Journal Commun. Soc.,} Mar. 2020.
\bibitem{Ochiai} H. Ochiai, P. Mitran and H. Vincent Poor, ``Capacity-approaching polar
codes with long codewords and successive cancellation decoding
based on improved Gaussian approximation," {\it IEEE Trans. on
Commun.} vol. 69, no. 1, pp. 31-43, 2021.

\bibitem{LYH23} G. Li, M. Ye and S. Hu, ``Adjacent-bits-swapped polar codes: A new code
construction to speed up polarization," {\it IEEE Trans. Inf. Theory,} vol. 69, no. 4, pp. 2269-2299, Apr. 2023.

\bibitem{Liu23} L. Liu, W. Yuan, Z. Liang, X. Ma and Z. Zhu, ``Construction of polar codes based on memetic algorithm," {\it IEEE Trans. on Emerging Topics in Computational Intelligence,} vol. 7, no. 5, pp.1539-1553, Oct. 2023.

\bibitem{Ma24} X. Yao and X. Ma, ``A balanced tree approach to construction
of length-flexible polar codes,"
{\it IEEE Trans. on
Commun.} vol. 72, no. 2, pp. 665-674, Feb. 2024.
 \end{thebibliography}
\end{document}